\documentclass[a4paper,12pt]{article}
\usepackage{verbatim}
\usepackage{graphicx}
\usepackage{epsfig}
\usepackage{amsmath,amssymb}
\usepackage{latexsym} 

\usepackage{bbding} 
\usepackage[T1]{fontenc} 
\usepackage{bbold} 
\usepackage{caption}
\usepackage{subcaption}

\newcommand{\floor}[1]{\lfloor{#1}\rfloor}
\newcommand{\ceil}[1]{\lceil{#1}\rceil}

\newcommand{\bw}{{\mathop {\rm bw}}}
\newcommand{\cp}{{\mathop {\rm cp}}}

\newcommand{\cl}{{\mathop {\rm cl}}}

\newcommand{\lm}{{\mathop {\rm lm}}}
\newcommand{\lp}{{\mathop {\rm lp}}}
\newcommand{\nd}{{\mathop {\rm nd}}}
\newcommand{\nx}{{\mathop {\rm nx}}}
\newcommand{\pv}{{\mathop {\rm pv}}}
\newcommand{\tw}{{\mathop {\rm tw}}}
\newcommand{\gm}{{\mathop {\rm gm}}}
\newcommand{\cm}{{\mathop {\rm cm}}}

\newcommand{\LT}{{\mathop{\rm LT}}}

\newcommand{\upDDG}{{\mathit up}{\mathop{\rm DDG}}}
\newcommand{\lowDDG}{{\mathit low}{\mathop{\rm DDG}}}
\newcommand{\dist}{{\mathop{\rm dist}}}
\newcommand{\reach}{{\mathop{\rm reach}}}

\newcommand{\ins}{{\mathop {\rm ins}}}

\newcommand{\AAA}{{\mathcal{A}}}

\newtheorem{lemma}{Lemma}
\newtheorem{theorem}{Theorem}
\newenvironment{proof}{\noindent {\bf Proof:~}}{$\Box$ \medskip}

\begin{document}

\begin{center}
{\Large Near-Linear Time Constant-Factor Approximation Algorithm

for Branch-Decomposition of Planar Graphs\footnote{A preliminary version of 
this paper appeared in the Proceedings of the 40th International Workshop on 
Graph-Theoretic Concepts in Computer Science (WG2014) \cite{GX14}.}
}
\vskip 0.2in

Qian-Ping Gu and Gengchun Xu

School of Computing Science, Simon Fraser University\\
Burnaby BC Canada V5A1S6\\
qgu@cs.sfu.ca,gxa2@sfu.ca  
\end{center}

\noindent
{\bf Abstract:}
We give an algorithm which for an input planar graph $G$ of $n$ vertices and 
integer $k$, in $\min\{O(n\log^3n),O(nk^2)\}$ time either constructs a 
branch-decomposition of $G$ with width at most $(2+\delta)k$, $\delta>0$ is a 
constant, or a $(k+1)\times \ceil{\frac{k+1}{2}}$ cylinder minor of $G$ implying 
$\bw(G)>k$, $\bw(G)$ is the branchwidth of $G$. This is the first $\tilde{O}(n)$ 
time constant-factor approximation for branchwidth/treewidth and largest 
grid/cylinder minors of planar graphs and improves the previous 
$\min\{O(n^{1+\epsilon}),O(nk^2)\}$ ($\epsilon>0$ is a constant) time 
constant-factor approximations. For a planar graph $G$ and $k=\bw(G)$, a 
branch-decomposition of width at most $(2+\delta)k$ and a $g\times \frac{g}{2}$ 
cylinder/grid minor with $g=\frac{k}{\beta}$, $\beta>2$ is constant, can be 
computed by our algorithm in $\min\{O(n\log^3n\log k),O(nk^2\log k)\}$ time.

\noindent
{\bf Key words:} Branch-/tree-decompositions, grid minor, planar graphs, 
approximation algorithm.

\section{Introduction}
\label{sec-intro}

The notions of branchwidth and branch-decomposition introduced by Robertson 
and Seymour \cite{RS91} in relation to the notions of treewidth and
tree-decomposition have important algorithmic applications. The branchwidth 
$\bw(G)$ and the treewidth $\tw(G)$ of graph $G$ are linearly related: 
$\max\{\bw(G),2\}\leq \tw(G)+1\leq \max\{\floor{\frac{3}{2}\bw(G)},2\}$ for 
every $G$ with more than one edge, and there are simple translations between 
branch-decompositions and tree-decompositions that meet the linear relations 
\cite{RS91}. A graph $G$ of small branchwidth (treewidth) admits efficient 
algorithms for many NP-hard problems \cite{ALS91,Bod93}. These algorithms first 
compute a branch-/tree-decomposition of $G$ and then apply a dynamic programming 
algorithm based on the decomposition to solve the problem. The dynamic programming 
step usually runs in polynomial time in the size of $G$ and exponential time in 
the width of the branch-/tree-decomposition computed.

Deciding the branchwidth/treewidth and computing a branch-/tree-decomposition of 
minimum width have been extensively studied. For an arbitrary graph $G$ of $n$
vertices, the following results have been known: Given an integer $k$, it is 
NP-complete to decide whether $\bw(G)\leq k$ \cite{ST94} ($\tw(G)\leq k$ 
\cite{ACP87}). If $\bw(G)$ ($\tw(G)$) is upper-bounded by a constant then both 
the decision problem and the optimal decomposition problem can be solved in 
$O(n)$ time \cite{BT97,Bod96}. However, the linear time algorithms are mainly of 
theoretical importance because the constant behind the Big-Oh is huge. The best 
known polynomial time approximation factor is $O(\sqrt{\bw(G)})$ for branchwidth 
and $O(\sqrt{\log \tw(G)})$ for treewidth \cite{FHL08}. The best known exponential
time approximation factors are as follows: an algorithm giving a 
branch-decomposition of width at most $3\bw(G)$ in $2^{O(\bw(G))}n^2$ time 
\cite{RS95}; an algorithm giving a tree-decomposition of width at most $3\tw(G)+4$ 
in $2^{O(\bw(G))}n\log n$ time \cite{Bet13}; and an algorithm giving a 
tree-decomposition of width at most $5\tw(G)+4$ in $2^{O(\tw(G))}n$ time 
\cite{Bet13}. By the linear relation between the branchwidth and treewidth,
the algorithms for tree-decompositions are also algorithms of same approximation
factors for branch-decompositions, while from a branchwidth approximation $\alpha$,
a treewidth approximation $1.5\alpha$ can be obtained.

Better results have been known for planar graphs $G$. Seymour and Thomas show 
that whether $\bw(G)\leq k$ can be decided in $O(n^2)$ time and an optimal 
branch-decomposition of $G$ can be computed in $O(n^4)$ time \cite{ST94}. Gu and
Tamaki improve the $O(n^4)$ time for the optimal branch-decomposition to $O(n^3)$
\cite{GT08}. By the linear relation between the branchwidth and treewidth, the
above results imply polynomial time $1.5$-approximation algorithms for the
treewidth and optimal tree-decomposition of planar graphs. It is open whether
deciding $\tw(G)\leq k$ is NP-complete or polynomial time solvable for planar 
graphs $G$.

Fast algorithms for computing small width branch-/tree-decompositions of planar 
graphs have received much attention as well. Tamaki gives an $O(n)$ time heuristic 
algorithm for branch-decomposition \cite{Tamaki03}. Gu and Tamaki give an algorithm 
which for an input planar graph $G$ of $n$ vertices and integer $k$, either 
constructs a branch-decomposition of $G$ with width at most $(c+1+\delta)k$ or 
outputs $\bw(G)>k$ in $O(n^{1+\frac{1}{c}})$ time, where $c$ is any fixed positive 
integer and $\delta>0$ is any constant \cite{GT11}. By this algorithm and 
a binary search, a branch-decomposition of width at most $(c+1+\delta)k$ can be 
computed in $O(n^{1+\frac{1}{c}}\log k)$ time, $k=\bw(G)$. Kammer and Tholey give 
an algorithm which for input $G$ and $k$, either constructs a tree-decomposition
of $G$ with width $O(k)$ or outputs $\tw(G)>k$ in $O(nk^3)$ time \cite{KT12,KT13}. 
The time complexity of the algorithm is improved to $O(nk^2)$ recently \cite{KT15}. 
This implies that a tree-decomposition of width $O(k)$ can be computed in
$O(nk^2\log k)$ time, $k=\tw(G)$. Computational study on branch-decomposition
can be found in \cite{BG08,BGMTY08,BGZ15,Hicks2005,Hicks2005-2,Smith2012,Tamaki03}.
Fast constant-factor approximation algorithms for branch-/tree-decompositions
of planar graphs have important applications such as that in shortest distance
oracles in planar graphs \cite{MS12}.

Grid minor of graphs is another notion in graph minor theory \cite{RST94}.
A $k\times k$ grid is a Cartesian product of two paths, each on $k$ vertices.
For a graph $G$, let $\gm(G)$ be the largest integer $k$ such that $G$ has
a $k\times k$ grid as a minor. Computing a large grid minor of a graph is
important in algorithmic graph minor theory and bidimensionality theory
\cite{DH05,DFHT05,RST94}. It is shown in \cite{RST94} that
$\gm(G)\leq \bw(G)\leq 4\gm(G)$ for planar graphs. Gu and Tamaki improve the
linear bound $\bw(G)\leq 4\gm(G)$ to $\bw(G)\leq 3\gm(G)$ and show that for any
$a<2$, $\bw(G)\leq a\gm(G)$ does not hold for planar graphs \cite{GT12}. 
Other studies on grid minor size and branchwidth/treewidth of planar graphs 
can be found in \cite{Bod08,Grigoriev11}. The upper bound $\bw(G)\leq 3\gm(G)$ 
is a consequence of a result on cylinder minors. A $k\times h$ cylinder is a 
Cartesian product of a cycle on $k$ vertices and a path on $h$ vertices. 
For a graph $G$, let $\cm(G)$ be the largest integer $k$ such that $G$ has a 
$k\times \ceil{\frac{k}{2}}$ cylinder as a minor. It is shown in \cite{GT12} 
that $\cm(G)\leq \bw(G)\leq 2\cm(G)$ for planar graphs. The $O(n^{1+\frac{1}{c}})$ 
time algorithm in \cite{GT11} actually constructs a branch-decomposition of $G$ 
with width at most $(c+1+\delta)k$ or a
$(k+1)\times \ceil{\frac{k+1}{2}}$ cylinder minor.

We propose an $\tilde{O}(n)$ time constant-factor approximation algorithm for
branch-/tree-decompositions of planar graphs. Our main result is as follows.
\begin{theorem} There is an algorithm which given a planar graph $G$ of $n$
vertices and an integer $k$, in $\min\{O(n\log^3n),O(nk^2)\}$ time either 
constructs a branch-decomposition of $G$ with width at most $(2+\delta)k$, 
$\delta>0$ is a constant, or a $(k+1)\times \ceil{\frac{k+1}{2}}$ cylinder 
minor of $G$.
\label{theo-1}
\end{theorem}
Since a $(k+1)\times \ceil{\frac{k+1}{2}}$ cylinder has branchwidth at least $k+1$
\cite{GT12}, a cylinder minor given in Theorem~\ref{theo-1} implies $\bw(G)>k$.

By the linear relation between the branchwidth and treewidth, Theorem~\ref{theo-1} 
implies an algorithm which for an input planar graph $G$ and integer $k$, in 
$\min\{O(nk+n\log^3n),O(nk^2)\}$ time constructs a tree-decomposition of $G$ with 
width at most $(3+\delta)k$ or outputs $\tw(G)>k$. For a planar graph $G$ and 
$k=\bw(G)$, by Theorem~\ref{theo-1} and a binary search, a branch-decomposition 
of $G$ with width at most $(2+\delta)k$ can be computed in 
$\min\{O(n\log^3n\log k),O(nk^2\log k)\}$ time. This improves the previous result 
of a branch-decomposition of width at most $(c+1+\delta)k$ in 
$O(n^{1+\frac{1}{c}}\log k)$ time \cite{GT11}. Similarly, for a planar graph $G$ 
and $k=\tw(G)$, a tree-decomposition of width at most $(3+\delta)k$ can be 
computed in $\min\{O(nk+n\log^3n\log k),O(nk^2\log k)\}$ time. Kammer and Tholey
give an algorithm which computes a tree-decomposition of $G$ with width at most 
$48k+13$ in $O(nk^3\log k)$ time or with width at most $(9+\delta)k+9$ in 
$O(n\min\{\frac{1}{\delta},k\}k^3\log k)$ time ($0<\delta<1$) \cite{KT12,KT13}. 
Recently, Kammer and Tholey give an algorithm for computing weighted treewidth for 
vertex weighted planar graphs \cite{KT15}. Applying this algorithm to planar graph 
$G$, a tree-decomposition of $G$ with width at most $(15+\delta)k+O(1)$ can be 
computed in $O(nk^2\log k)$ time. This improves the result of \cite{KT12,KT13}. 
Our $O(nk^2\log k)$ time algorithm is an independent improvement over the result 
of \cite{KT12,KT13}\footnote{The $O(nk^2\log k)$ time algorithm in \cite{KT15} was 
announced in July 2015 while our our result was reported in March 2015 
\cite{GX15}.} and has a better approximation ratio than that of \cite{KT15}. Our 
algorithm can also be used to compute a $g\times \ceil{\frac{g}{2}}$ cylinder 
(grid) minor with $g=\frac{\bw(G)}{\beta}$, $\beta>2$ is a constant, and a 
$g\times g$ cylinder (grid) minor with $g=\frac{\bw(G)}{\beta}$, $\beta>3$ is a 
constant, of $G$ in $\min\{O(n\log^3n\log k),O(nk^2\log k)\}$ time. This improves 
the previous results of $g\times \ceil{\frac{g}{2}}$ with 
$g\geq \frac{\bw(G)}{\beta}$, $\beta>(c+1)$, and $g\times g$ with 
$g\geq \frac{\bw(G)}{\beta}$, $\beta>(2c+1)$, in $O(n^{1+\frac{1}{c}}\log k)$ 
time. As an application, our algorithm removes a bottleneck in the work of 
\cite{MS12} for computing a shortest path oracle and reduces its preprocessing 
time in Theorem 6.1 from $O(n^{1+\frac{1}{c}}\log k\log n+S\log ^2n)$ to 
$O(\min\{O(n\log^4n\log k),O(nk^2\log n\log k)\}+S\log ^2n)$.

Our algorithm for Theorem~\ref{theo-1} uses the approach in the previous work 
of \cite{GT11} described below. Given a planar graph $G$ and integer $k$, let 
${\cal Z}$ be the set of biconnected components of $G$ with a normal distance 
(a definition is given in the next section) $h=ak$, $a>0$ is a constant, from a 
selected edge $e_0$ of $G$. For each $Z\in {\cal Z}$, a minimum vertex cut set 
$\partial(A_Z)$ which partitions $E(G)$ into edge subsets $A_Z$ and 
$\overline{A}_Z=E(G)\setminus A_Z$ is computed such that $Z\subseteq A_Z$ and 
$e_0\in \overline{A}_Z$, that is, $\partial(A_Z)$ separates $Z$ and $e_0$. If 
$|\partial(A_Z)|>k$ for some $Z\in {\cal Z}$ then $\bw(G)>k$ is concluded. 
Otherwise, a branch-decomposition of graph $H$ obtained from $G$ by removing all 
$A_Z$ is constructed. For each subgraph $G[A_Z]$ induced by $A_Z$, a 
branch-decomposition is constructed or $\bw(G[A_Z])>k$ is concluded recursively. 
Finally, a branch-decomposition of $G$ with width $O(k)$ is constructed from the 
branch-decomposition of $H$ and those of $G[A_Z]$ or $\bw(G)>k$ is concluded. 

The algorithm in \cite{GT11} computes a minimum vertex cut set $\partial(A_Z)$ 
for every $Z\in {\cal Z}$ in all recursive steps in $O(n^{1+\frac{1}{c}})$ time. 
Our main idea for proving Theorem~\ref{theo-1} is to find a minimum vertex cut 
set $\partial(A_Z)$ for every $Z\in {\cal Z}$ more efficiently based on recent 
results for computing minimum face separating cycles and vertex cut sets in 
planar graphs. Borradaile et al. give an algorithm which in $O(n\log^4 n)$ 
time computes an oracle for the all pairs minimum face separating cycle problem 
in a planar graph $G$ \cite{BSW13}. The time for computing the oracle is 
further improved to $O(n\log^3 n)$ \cite{Bet14}. For any pair of faces $f$ and 
$g$ in $G$, the oracle in $O(|C|)$ time returns a minimum $(f,g)$-separating 
cycle $C$ ($C$ cuts the sphere on which $G$ is embedded into two regions, one 
contains $f$ and the other contains $g$). By this result, we show that a minimum 
vertex cut set $\partial(A_Z)$ for every $Z\in {\cal Z}$ in all recursive steps 
can be computed in $O(n\log^3 n)$ time and get the next result.
\begin{theorem} There is an algorithm which given a planar graph $G$ of $n$ 
vertices and an integer $k$, in $O(n\log^3 n)$ time either constructs a 
branch-decomposition of $G$ with width at most $(2+\delta)k$ or a 
$(k+1)\times \ceil{\frac{k+1}{2}}$ cylinder minor of $G$, where $\delta>0$
is a constant.
\label{theo-2}
\end{theorem}

For an input $G$ and integer $k$, Kammer and Tholey give an algorithm which in 
$O(nk^3)$ time constructs a tree-decomposition of width $O(k)$ or outputs 
$\tw(G)>k$ as follows \cite{KT12,KT13}: Convert $G$ into an almost triangulated 
planar graph $\hat{G}$. Use {\em crest separators} to decompose $\hat{G}$ into 
pieces (subgraphs), each piece contains one component (called {\em crest} with a 
normal distance $k$ from a selected set of edges called {\em coast}). For each 
crest compute a vertex cut set of size at most $3k-1$ to separate the crest from 
the coast. If such a vertex cut set can not be found for some crest then the 
algorithm concludes $\tw(\hat{G})>k$. Otherwise, the algorithm computes a 
tree-decomposition for the graph $\hat{H}$ obtained by removing all crests from 
$\hat{G}$ and works on each crest recursively. Finally, the algorithm constructs 
a tree-decomposition of $\hat{G}$ from the tree-decomposition of $\hat{H}$ and 
those of crests. 

To get an $O(nk^2)$ time algorithm for Theorem~\ref{theo-1}, we apply the ideas 
of triangulating $G$ and crest separators in \cite{KT12,KT13} to decompose 
$\hat{G}$ into pieces, each piece having one component (crest) $Z\in {\cal Z}$. 
Instead of finding a vertex cut set of size at most $3k-1$ for each crest, we apply 
the minimum face separating cycle to find a minimum vertex cut set $\partial(A_Z)$ 
in each piece. We show that either a vertex cut set $\partial(A_Z)$ with 
$|\partial(A_Z)|\leq k$ for every $Z\in {\cal Z}$ in all recursive steps or a 
$(k+1)\times \ceil{\frac{k+1}{2}}$ cylinder minor can be computed in $O(nk^2)$ 
time and get the result below.
\begin{theorem} There is an algorithm which given a planar graph $G$ of $n$ 
vertices and an integer $k$, in $O(nk^2)$ time either constructs a 
branch-decomposition of $G$ with width at most $(2+\delta)k$ or a 
$(k+1)\times \ceil{\frac{k+1}{2}}$ cylinder minor of $G$, where $\delta>0$ 
is a constant.
\label{theo-3}
\end{theorem}
Theorem~\ref{theo-1} follows from Theorems~\ref{theo-2} and \ref{theo-3}.

The next section gives the preliminaries of the paper. We prove Theorems 
\ref{theo-2} and \ref{theo-3} in Sections 3 and 4, respectively. The final 
section concludes the paper.

\section{Preliminaries}
\label{sec-prel}

It is convenient to view a vertex cut set $\partial(A_Z)$ in a graph as an edge 
in a hypergraph in some cases. A hypergraph $G$ consists of a set $V(G)$ of 
vertices and a set $E(G)$ of edges, each edge is a subset of $V(G)$ with at least 
two elements. A hypergraph $G$ is a graph if for every $e\in E(G)$, $e$ has two 
elements. 
For a subset $A\subseteq E(G)$, we denote $\cup_{e\in A} e$ by $V(A)$ and 
denote $E(G)\setminus A$ by $\overline{A}$. For $A\subseteq E(G)$, the pair 
$(A,\overline{A})$ is a {\em separation} of $G$ and we denote by $\partial(A)$ 
the vertex set $V(A)\cap V(\overline{A})$. The {\em order} of separation 
$(A,\overline{A})$ is $|\partial(A)|$. A hypergraph $H$ is a subgraph of $G$ if 
$V(H)\subseteq V(G)$ and $E(H)\subseteq E(G)$. For $A\subseteq E(G)$ and 
$W\subseteq V(G)$, we denote by $G[A]$ and $G[W]$ the subgraphs of $G$ induced
by $A$ and $W$, respectively. For a subgraph $H$ of $G$, we denote 
$G[E(G)\setminus E(H)]$ by $G\setminus H$.

A walk in graph $G$ is a sequence of edges $e_1,e_2,...,e_k$, where
$e_i=\{v_{i-1},v_i\}$. We call $v_0$ and $v_k$ the {\em end vertices} and other
vertices the {\em internal vertices} of the walk. A walk is a path if all vertices
in the walk are distinct. A walk is a cycle if it has at least three vertices,
$v_0=v_k$ and $v_1,...,v_k$ are distinct. A graph is {\em weighted} if each edge 
of the graph is assigned a weight. Unless otherwise stated, a graph is unweighted.
The {\em length} of a walk in a graph is the number of edges in the walk. The 
length of a walk in a weighted graph is the sum of the weights of the edges in 
the walk.

The notions of branchwidth and branch-decomposition are introduced by Robertson 
and Seymour \cite{RS91}. A {\em branch-decomposition} of hypergraph $G$ is a pair
$(\phi,T)$ where $T$ is a ternary tree and $\phi$ is a bijection from the set of 
leaves of $T$ to $E(G)$. We refer the edges of $T$ as links and the vertices of 
$T$ as nodes. Consider a link $e$ of $T$ and let $L_1$ and $L_2$ denote the sets 
of leaves of $T$ in the two respective subtrees of $T$ obtained by removing $e$. 
We say that the separation $(\phi(L_1),\phi(L_2))$ is induced by this link $e$ of 
$T$. We define the width of the branch-decomposition $(\phi,T)$ to be the largest
order of the separations induced by links of $T$. The {\em branchwidth} of $G$, 
denoted by $\bw(G)$, is the minimum width of all branch-decompositions of $G$. 
In the rest of this paper, we identify a branch-decomposition $(\phi,T)$ with the
tree $T$, leaving the bijection implicit and regarding each leaf of $T$ as a 
edge of $G$.

Let $\Sigma$ be a sphere. For an open segment $s$ homeomorphic to
$\{x|0<x<1\}$ in $\Sigma$, we denote by $\cl(s)$ the closure of $s$.
A planar embedding of a graph $G$ is a mapping 
$\rho: V(G)\cup E(G) \to \Sigma\cup 2^{\Sigma}$ such that
\begin{itemize}
\item for $u\in V(G)$, $\rho(u)$ is a point of $\Sigma$, and for distinct 
$u,v\in V(G)$, $\rho(u)\neq \rho(v)$;
\item for each edge $e=\{u,v\}\in E(G)$, $\rho(e)$ is an open segment in 
$\Sigma$ with $\rho(u)$ and $\rho(v)$ the two end points in
$\cl(\rho(e))\setminus \rho(e)$; and 
\item for distinct $e_1,e_2\in E(G)$,
$\cl(\rho(e_1))\cap \cl(\rho(e_2))=\{\rho(u)|u\in e_1\cap e_2\}$.
\end{itemize}
A graph $G$ is planar if it has a planar embedding $\rho$, and $(G,\rho)$ is 
called a plane graph. We may simply use $G$ to denote the plane graph $(G,\rho)$, 
leaving the embedding $\rho$ implicit. For a plane graph $G$, each connected 
component of $\Sigma\setminus (\cup_{e\in E(G)} \cl(\rho(e)))$ is a face 
of $G$. 
We denote by $V(f)$ and $E(f)$ the set of vertices and the set of edges incident 
to face $f$, respectively. We say that face $f$ is bounded by the edges of $E(f)$.

A graph $G$ of at least three vertices is {\em biconnected} if for any pairwise 
distinct vertices $u,v,w\in V(G)$, there is a path of $G$ between $u$ and $v$ that 
does not contain $w$. Graph $G$ of a single vertex or a single edge is 
(degenerated) biconnected. A {\em biconnected component} of $G$ is a maximum
biconnected subgraph of $G$. It suffices to prove Theorems~\ref{theo-2} and
\ref{theo-3} for a biconnected $G$ because if $G$ is not biconnected, the 
problems of finding branch-decompositions and cylinder minors of $G$ can be 
solved individually for each biconnected component.

For a plane graph $G$, a curve $\mu$ on $\Sigma$ is {\em normal} if $\mu$ does not 
intersect any edge of $G$. The length of a normal curve $\mu$ is the number of 
connected components of $\mu \setminus \bigcup_{v \in V(G)} \{\rho(v)\}$.
For vertices $u,v \in V(G)$, the {\em normal distance} $\nd_G(u,v)$ 
is defined as the shortest length of a normal curve between $\rho(u)$ and 
$\rho(v)$. The {\em normal distance} between two vertex-subsets 
$U,W \subseteq V(G)$ is defined as $\nd_G(U,W) = \min_{u\in U, v\in W} \nd_G(u,v)$.
We also use $\nd_G(U,v)$ for $\nd_G(U,\{v\})$ and $\nd_G(u,W)$ for 
$\nd_G(\{u\},W)$.

A {\em noose} of $G$ is a closed normal curve on $\Sigma$ that does not
intersect with itself. A noose $\nu$ of $G$ 
separates $\Sigma$ into two open regions $R_1$ and $R_2$ and induces a separation 
$(A,\overline{A})$ of $G$ with $A=\{e\in E(G)\mid \rho(e)\subseteq R_1\}$ and
$\overline{A}=\{e\in E(G)\mid \rho(e)\subseteq R_2\}$. We also say $\nu$ induces
edge subset $A$ ($\overline{A}$). A separation (resp. an edge subset) of $G$ is 
called {\em noose-induced} if there is a noose which induces the separation 
(resp. edge subset). A noose $\nu$ separates two edge subsets $A_1$ and $A_2$ 
if $\nu$ induces a separation $(A,\overline{A})$ with $A_1\subseteq A$ and 
$A_2\subseteq \overline{A}$. We also say that the noose induced subset $A$ 
separates $A_1$ and $A_2$.

For plane graph $G$ and a noose $\nu$ induced $A\subseteq E(G)$, we denote by 
$G|A$ the plane hypergraph obtained by replacing all edges of $A$ with edge 
$\partial(A)$ (i.e., $V(G|A)=(V(G)\setminus V(A))\cup \partial(A)$ and 
$E(G|A)=(E(G)\setminus A)\cup \{\partial(A)\}$). An embedding of $G|A$ can 
be obtained from $G$ with $\rho(\partial(A))$ an open disk (homeomorphic to 
$\{(x,y)|x^2+y^2<1\}$) which is the open region separated by $\nu$ and contains 
$A$. For a collection $\AAA=\{A_1,..,A_r\}$ of mutually disjoint noose induced
edge-subsets of $G$, $(..(G|A_1)|..)|A_r$ is denoted by $G|\AAA$. 

\section{$O(n\log^3n)$ time algorithm}
\label{sec-alg1}

We give an algorithm to prove Theorem~\ref{theo-2}. Our algorithm follows the
approach of the work in \cite{GT11}. Let $G$ be a plane graph (hypergraph) of $n$ 
vertices, $e_0$ be an arbitrary edge of $G$ and $k,h>0$ be integers. We first try 
to separate $e_0$ and the subgraph of $G$ induced by the vertices with the normal 
distance at least $h$ from $e_0$. Since the subgraph may not be biconnected, 
let ${\cal Z}$ be the set of biconnected components of $G$ such that for each 
$Z\in {\cal Z}$, $\nd_G(e_0,V(Z))=h$. For each $Z\in {\cal Z}$, our algorithm 
computes a minimum noose induced subset $A_Z$ separating $Z$ and $e_0$. If for 
some $Z\in {\cal Z}$, $|\partial(A_Z)|>k$ then the algorithm constructs a 
$(k+1)\times h$ cylinder minor of $G$ in $O(n)$ time by Lemma~\ref{lem:minor} 
proved in \cite{GT11}. Otherwise, a set $\AAA$ of noose induced subsets with the 
following properties is computed: (1) for every $A_Z\in \AAA$, 
$|\partial(A_Z)|\leq k$, (2) for every $Z\in {\cal Z}$, there is an $A_Z\in \AAA$ 
which separates $Z$ and $e_0$ and (3) for distinct $A_Z,A_{Z'}\in \AAA$, 
$A_Z\cap A_Z'=\emptyset$. Such an $\AAA$ is called a {\em good-separator} for 
${\cal Z}$ and $e_0$.

\begin{lemma} \cite{GT11}
Given a plane graph $G$ and integers $k,h>0$, let $A_1$ and $A_2$ be edge subsets
of $G$ satisfying the following conditions: (1) each of separations
$(A_1,\overline{A_1})$ and $(A_2,\overline{A_2})$ is noose-induced; (2) $G[A_2]$
is biconnected; (3) $\nd_G(V(\overline{A_1}),V(A_2))\geq h$; and (4) every noose of
$G$ that separates $\overline{A_1}$ and $A_2$ has length $>k$. Then $G$ has a
$(k+1)\times h$ cylinder minor and given $(G|\overline{A_1})|A_2$, such a
minor can be constructed in $O(|V(A_1\cap \overline{A_2})|)$ time.
\label{lem:minor}
\end{lemma}

Given a good-separator $\AAA$ for ${\cal Z}$ and $e_0$, our algorithm constructs 
a branch-decomposition of plane hypergraph $G|\AAA$ with width at most $k+2h$ by 
Lemma~\ref{lem:branchwidth} shown in \cite{GT12,Tamaki03}. For each $A_Z\in \AAA$, 
the algorithm computes a cylinder minor or a branch-decomposition for the plane 
hypergraph $G|\overline{A}_Z$ recursively. If a branch-decomposition of 
$G|\overline{A}_Z$ is found for every $A_Z\in \AAA$, the algorithm constructs a 
branch-decomposition of $G$ with width at most $k+2h$ from the 
branch-decomposition of $G|\AAA$ and those of $G|\overline{A}_Z$ by 
Lemma~\ref{lem:decompose} which is straightforward from the definitions of
branch-decompositions.


\begin{lemma}
\label{lem:branchwidth}
\cite{GT12,Tamaki03}
Let $k>0$ and $h>0$ be integers. Let $G$ be a plane hypergraph with each edge of
$G$ incident to at most $k$ vertices. If there is an edge $e_0$ such that for any
vertex $v$ of $G$, $\nd_G(e_0,v)\leq h$ then given $e_0$, a branch-decomposition
of $G$ with width at most $k+2h$ can be constructed in $O(|V(G)|+|E(G)|)$ time.
\end{lemma}
The upper bound $k+2h$ is shown in Theorem 3.1 in \cite{GT12}. The normal distance
in \cite{GT12} between a pair of vertices is twice of the normal distance in
this paper between the same pair of vertices. Tamaki gives a linear time algorithm
to construct a branch-decomposition of width at most $k+2h$ \cite{Tamaki03}.

The following lemma is straightforward from the definition of branch-decompositions 
and allows us to bound the width of the branch-decomposition of the whole graph.

\begin{lemma}
\label{lem:decompose}
Given a plane hypergraph $G$ and a noose-induced separation $(A,\overline{A})$ 
of $G$, let $T_A$ and $T_{\overline{A}}$ be branch-decompositions of 
$G|\overline{A}$ and $G|A$ respectively. Let $T_A+T_{\overline{A}}$ to be the 
tree obtained from $T_A$ and $T_{\overline{A}}$ by joining the link incident
to the leaf $\partial(A)$ in $T_A$ and the link incident to the leaf $\partial(A)$ 
in $T_{\overline{A}}$ into one link and removing the leaves $\partial(A)$. Then 
$T_A+T_{\overline{A}}$ is a branch-decomposition of $G$ with width 
$\max\{|\partial(A)|,k_A,k_{\overline{A}}\}$ where $k_A$ is the
width of $T_A$ and $k_{\overline{A}}$ is the width of
$T_{\overline{A}}$.
\end{lemma}

To make a concrete progress in each recursive step, the following technique in 
\cite{GT11} is used to compute $\AAA$. For a plane hypergraph $G$, a vertex 
subset $e_0$ of $G$ and an integer $d\geq 0$, let
\[
\reach_G(e_0,d)=\bigcup \{v\in V(G)|\nd_G(e_0,v) \leq d\} 
\]
denote the set of vertices of $G$ with the normal distance at most $d$ from set 
$e_0$. Let $\alpha>0$ be an arbitrary constant. For integer $k\geq 2$, let
$d_1=\ceil{\frac{\alpha k}{2}}$ and $d_2=d_1+\ceil{\frac{k+1}{2}}$. 
The {\em layer tree} $\LT(G,e_0)$ is defined as follows:
\begin{enumerate}
\item the root of the tree is $G$;
\item each biconnected component $X$ of $G[V(G)\setminus \reach_G(e_0,d_1-1)]$ 
is a node in level 1 of the tree and is a child of the root; and
\item each biconnected component $Z$ of $G[V(G)\setminus \reach_G(e_0,d_2-1)]$ is 
a node in level 2 of the tree and is a child of the biconnected component $X$
in level 1 that contains $Z$.
\end{enumerate}
For $h=d_2$, ${\cal Z}$ is the set of leaf nodes of $\LT(G,e_0)$ in level 2. For a 
node $X$ of $\LT(G,e_0)$ in level 1 that is not a leaf, let ${\cal Z}_X$ be the set 
of child nodes of $X$. It is shown in \cite{GT11} (in the proofs of Lemma 4.1) that 
for any $Z\in {\cal Z}_X$, if a minimum noose in the plane hypergraph 
$(G|\overline{X})|{\cal Z}_X$ separating $\{\partial(Z)\}$ and 
$\{\partial(\overline{X})\}$ has length $>k$ then $G$ has a 
$(k+1)\times \ceil{\frac{k+1}{2}}$ cylinder minor. From this, a good-separator 
$\AAA_X$ for ${\cal Z}_X$ and $\overline{X}$ can be computed in hypergraph
$(G|\overline{X})|{\cal Z}_X$, and the union of $\AAA_X$ for every $X$ gives a 
good-separator $\AAA$ for ${\cal Z}$ and $e_0$. 

Notice that if $Z$ is a single vertex then $Z$ will not be involved any further 
recursive step; and if $Z$ is a single edge then there is a noose of length
$2\leq k$ separating $\{\partial(Z)\}$ and $\{\partial(\overline{X})\}$, and
it is trivial to compute the branch-decomposition of $Z$. So we assume without 
loss of generality that each $Z\in {\cal Z}_X$ has at least three vertices. 

To compute $\AAA_X$, we convert $(G|\overline{X})|{\cal Z}_X$ to a weighted plane 
graph and compute a minimum noose induced subset $A_Z$ separating $Z\in {\cal Z}_X$ 
and $\overline{X}$ by finding a minimum face separating cycle in the weighted 
plane graph. We use the algorithm by Borradaile et al. \cite{Bet14} to compute the 
face separating cycles.

For each edge $\partial(Z)$ in $(G|\overline{X})|{\cal Z}_X$, let $\nu_Z$ be 
the noose which induces the separation $(Z,\overline{Z})$ in $G$. Then
$E_Z=\{\nu_X\setminus \rho(u)|u\in \partial(Z)\}$ is a set of open segments.
We first convert hypergraph $(G|\overline{X})|{\cal Z}_X$ into a plane graph $G_X$
as follows: Remove edge $\rho(\partial(\overline{X}))$ and
for each $Z\in {\cal Z}_X$, replace edge $\rho(Z)$ by the set of edges which 
are the segments in $E_Z$.

$G_X$ has a face which contains $\rho(\partial(\overline{X}))$ and we denote this
face by $f_{\overline{X}}$. Notice that 
$V(f_{\overline{X}})=\partial(\overline{X})$ and the edges of
$E(f_{\overline{X}})$ form a cycle because $X$ is biconnected. For each 
$Z\in {\cal Z}_X$, the embedding $\rho(\partial(Z))$ of edge $\partial(Z)$ 
becomes a face $f_Z$ in $G_X$ with $E(f_Z)=E_Z$. A face in $G_X$ which is not 
$f_{\overline{X}}$ or any of $f_Z$ is called a {\em natural face} in $G_X$. Next 
we convert $G_X$ to a weighted plane graph $H_X$ as follows: For each natural 
face $f$ in $G_X$ with $|V(f)|>3$, we add a new vertex $u_f$ and new edges 
$\{u_f,v\}$ in $f$ for every vertex $v$ in $V(f)$. Each new edge $\{u_f,v\}$ is 
assigned the weight $1/2$. Each edge of $G_X$ is assigned the weight 1. Notice 
that $|V(H_X)|=O(|V(G_X)|)$. 

For $Z\in {\cal Z}_X$, a minimum $(f_Z,f_{\overline{X}})$-separating cycle is a 
cycle separating $f_Z$ and $f_{\overline{X}}$ with the minimum length. A noose in 
$G_X$ is called a {\it natural noose} if it intersects only natural faces in $G_X$. 
It is shown (Lemma 5.1) in \cite{GT11} that for each $Z\in {\cal Z}_X$, a minimum 
natural noose in $G_X$ separating $E(f_Z)$ and $E(f_{\overline X})$ in $G_X$ is a 
minimum noose separating $\{\partial(Z)\}$ and $\{\partial(\overline{X})\}$ in 
$(G|\overline{X})|{\cal Z}_X$. By Lemma~\ref{cycle-noose} below, such a natural 
noose $\nu$ can be computed by finding a minimum $(f_Z,f_{\overline{X}})$-separating 
cycle $C$ in $H_X$. The subset $A_Z$ induced by noose $\nu$ in 
$(G|\overline{X})|{\cal Z}_X$ is also called {\em cycle $C$ induced subset}.

\begin{lemma}
Let $H_X$ be the weighted plane graph obtained from $G_X$. For any 
$(f_Z,f_{\overline{X}})$-separating cycle $C$ in $H_X$, there is a natural noose 
$\nu$ which separates $E(f_Z)$ and $E(f_{\overline{X}})$ in $G_X$ with the same 
length as that of $C$. For any minimum natural noose $\nu$ in $G_X$ separating 
$E(f_Z)$ and $E(f_{\overline{X}})$, there is a $(f_Z,f_{\overline{X}})$-separating 
cycle $C$ in $H_X$ with the same length as that of $\nu$.
\label{cycle-noose}
\end{lemma}
\begin{proof}
Let $C$ be a $(f_Z,f_{\overline{X}})$-separating cycle in $H_X$. For each edge 
$\{u,v\}$ in $C$ with $u,v\in V(G_X)$, $\{u,v\}$ is incident to a natural face
$f$ because $f_Z$ is not incident to $f_{\overline X}$ by
$\nd_{G_X}(V(\overline{X}),V(Z))=\ceil{\frac{k+1}{2}}$. We draw a simple curve 
with $u,v$ as its end points in face $f$. For each pair of edges $\{u,u_f\}$ and 
$\{u_f,v\}$ in $C$ with $u,v\in V(G_X)$ and $u_f\in V(H_X)\setminus V(G_X)$, we 
draw a simple curve with $u,v$ as its end points in the face $f$ of $G_X$ where 
the new added vertex $u_f$ is placed. Then the union of the curves form a natural 
noose $\nu$ which separates $E(f_Z)$ and $E(f_{\overline{X}})$ in $G_X$. Each of 
edge $\{u,v\}$ with $u,v\in V(G_X)$ is assigned weight 1. For a new added vertex 
$u_f$, each of edges $\{u,u_f\},\{u_f,v\}$ is assigned weight $1/2$. Therefore, 
the lengths of $\nu$ and $C$ are the same.

Let $\nu$ be a minimum natural noose separating $E(f_Z)$ and $E(f_{\overline{X}})$ 
in $G_X$. Then $\nu$ contains at most two vertices of $G_X$ incident to a same 
natural face of $G_X$, otherwise a shorter natural noose separating $E(f_Z)$ and 
$E_{\overline{X}}$ can be formed. The vertices on $\nu$ partitions $\nu$ into a 
set of simple curves such that at most one curve is drawn in each natural face 
of $G_X$. For a curve with the end points $u$ and $v$ in a natural face $f$, if 
$\{u,v\}$ is an edge of $G_X$ then we take $\{u,v\}$ in $H_X$ as a candidate, 
otherwise we take edges $\{u,u_f\},\{u_f,v\}$ in $H_X$ as candidates, where $u_f$ 
is the vertex added in $f$ in getting $H_X$. These candidates form a 
$(f_Z,f_{\overline{X}})$-separating cycle $C$ in $H_X$. Because each edge 
of $G_X$ is given weight 1 and each added edge is given weight $1/2$ in $H_X$, 
the lengths of $C$ and $\nu$ are the same.
\end{proof} 

We assume that for every pair of vertices $u,v$ in $H_X$, there 
is a unique shortest path between $u$ and $v$. This can be realized by 
perturbating the edge weight $w(e)$ of each edge $e$ in $H_X$ as follows.
Assume that the edges in $H_X$ are $e_1,...e_m$. For each edge $e_i$, let 
$w'(e_i)=w(e_i)+\frac{1}{2^{i+1}}$. Then it is easy to check that for any pair of 
vertices $u$ and $v$ in $H_X$, there is a unique shortest path between $u$ and 
$v$ w.r.t. to $w'$; and the shortest path between $u$ and $v$ w.r.t. $w'$ is a
shortest path between $u$ and $v$ w.r.t. $w$.

For a plane graph $G$, a {\em minimum cycle base tree} (MCB tree) introduced
in \cite{BSW13} is an edge-weighted tree $\tilde{T}$ such that
\begin{itemize}
\item there is a bijection from the faces of $G$ to the nodes of $\tilde{T}$;
\item removing each edge $e$ from $\tilde{T}$ partitions $\tilde{T}$ into two 
subtrees $\tilde{T}_1$ and $\tilde{T}_2$; this edge $e$ corresponds to a cycle 
which separates every pair of faces $f$ and $g$ with $f$ in $\tilde{T}_1$ and 
$g$ in $\tilde{T}_2$; and
\item for any distinct faces $f$ and $g$, the minimum-weight edge on the
unique path between $f$ and $g$ in $\tilde{T}$ has weight equal to the length
of a minimum $(f,g)$-separating cycle.
\end{itemize}

The next lemma gives the running time for computing a MCB tree of a
plane graph and that for obtaining a cycle from the MCB tree.
\begin{lemma}\cite{Bet14}
Given a plane graph $G$ of $n$ vertices with positive edge weights, a MCB tree 
of $G$ can be computed in $O(n\log^3n)$ time. Further, for any distinct faces $f$ 
and $g$ in $G$, given a minimum weight edge in the path between $f$ and $g$ in
the MCB tree, a minimum $(f,g)$-separating cycle $C$ can be obtained in $O(|C|)$ 
time, $|C|$ is the number of edges in $C$.
\label{lem-mcb}
\end{lemma}

Using Lemma~\ref{lem-mcb} for computing a MCB tree $\tilde{T}$ of $H_X$ and thus
$\AAA_X$, our algorithm is summarized in Procedure Branch-Minor below. In the 
procedure, $U$ is a noose induced edge subset and initially $U=\{e_0\}$.

\noindent{\bf Procedure} Branch-Minor($G|U$)\\
{\bf Input:} A biconnected plane hypergraph $G|U$ with $\partial(U)$ specified,
$|\partial(U)|\leq k$ and every other edge has two vertices. \\
{\bf Output:} Either a branch-decomposition of $G|U$ of width at most $k+2h$,
$h=d_2$, or a $(k+1)\times \ceil{\frac{k+1}{2}}$ cylinder minor of $G$.
\begin{enumerate}
\item If $\nd_{G|U}(\partial(U),v)\leq h$ for every $v\in V(G|U)$ then apply
Lemma~\ref{lem:branchwidth} to find a branch-decomposition of $G|U$. Otherwise,
proceed to the next step.

\item Compute the layer tree $\LT(G|U,\partial(U))$.

For every node $X$ of $\LT(G|U,\partial(U))$ in level 1 that is not a leaf, 
compute $\AAA_X$ as follows:
\begin{enumerate}
\item Compute $H_X$ from $(G|\overline{X})|{\cal Z}_X$.
\item Compute a MCB tree $\tilde{T}$ of $H_X$ by Lemma~\ref{lem-mcb}.
\item Find a face $f_Z$, $Z\in {\cal Z}_X$, in $\tilde{T}$ by a breadth first
search from $f_{\overline X}$ such that the path between $f_Z$ and 
$f_{\overline X}$ in $\tilde{T}$ does not contain $f_{Z'}$ for any 
$Z'\in {\cal Z}_X$ with $Z'\neq Z$. Find the minimum weight edge $e_Z=\{u,v\}$ 
in the path between $f_Z$ and $f_{\overline X}$, and the cycle $C$ from edge 
$e_Z$.

If $C$ has length $>k$, then compute a $(k+1)\times \ceil{\frac{k+1}{2}}$ 
cylinder minor by Lemma~\ref{lem:minor} and terminate.

Otherwise, compute the cycle $C$ induced subset $A_Z$ and include $A_Z$ to
$\AAA_X$. For each node $f$ of $\tilde{T}$, if edge $e_Z$ is in the path between 
$f$ and $f_{\overline X}$ in $\tilde{T}$ then delete $f$ from $\tilde{T}$.

Repeat the above until $\tilde{T}$ does not contain any $f_Z$ for 
$Z\in {\cal Z}_X$.

\end{enumerate}
Let $\AAA=\cup_{X: \mbox{level 1 node}} {\cal A}_X$ and proceed to the next step.

\item For each $A\in \AAA$, 
call Branch-Minor($G|\overline{A}$) to construct a branch-decomposition $T_{A}$
or a cylinder minor of $G|\overline{A}$.

If a branch-decomposition $T_A$ is found for every $A\in \AAA$,
Lemma~\ref{lem:branchwidth} is applied to $(G|U)|\AAA$ to construct a
branch-decomposition $T_0$ of $(G|U)|\AAA$ and Lemma~\ref{lem:decompose} is
used to combine these branch-decompositions $T_A$, $A\in \AAA$, and $T_0$
into a branch-decomposition $T$ of $G|U$ and return $T$.
\end{enumerate}

Now we prove Theorem~\ref{theo-2} which is re-stated below.
\begin{theorem} There is an algorithm which given a planar graph $G$ of $n$
vertices and an integer $k$, in $O(n\log^3n)$ time either constructs a 
branch-decomposition of $G$ with width at most $(2+\delta)k$, $\delta>0$ is a 
constant, or a $(k+1)\times \ceil{\frac{k+1}{2}}$ cylinder minor of $G$.
\label{theo-4}
\end{theorem}
\begin{proof}
The input hypergraph $G|\overline{A}$ of our algorithm in each recursive step 
for $A\in \AAA$ is biconnected. For the $\AAA_X$ computed in Step 2, obviously 
(1) for every $A_Z\in \AAA_X$, $|\partial(A_Z)|\leq k$; (2) due to the way we 
find the cycles from the MCB tree, for every $Z\in {\cal Z}_X$, there is exactly 
one noose-induced subset $A_Z\in \AAA_X$ separating $Z$ and $\overline{X}$; and 
(3) from the unique shortest path in $H_X$, for distinct $A_Z,A_{Z'}\in \AAA_X$,
$A_Z\cap A_{Z'}=\emptyset$. Therefore, $\AAA_X$ is a good-separator 
for ${\cal Z}_X$ and $\overline{X}$. From this, $\AAA$ is a good separator
for ${\cal Z}$ and $U$ and our algorithm computes a branch-decomposition 
or a $(k+1)\times \ceil{\frac{k+1}{2}}$ cylinder minor of $G$. The width of the 
branch-decomposition computed is at most 
\[
k+2h = k+2(d_1+\ceil{\frac{k+1}{2}}) \leq 
k+2(\ceil{\frac{\alpha k}{2}})+(k+2) \leq (2+\delta)k,
\]
where $\delta$ is the smallest constant with $\delta k\geq \alpha k+4$.

Let $M,m_x,m$ be the numbers 
of edges in $G[\reach_{G|U}(\partial(U),d_2)],(G|\overline{X})|{\cal Z}_X,H_X$, 
respectively. Then $m=O(m_x)$. In Step 2, the layer tree $\LT(G|U,\partial(U))$ 
can be computed in $O(M)$ time. For each level 1 node $X$, it takes $O(m)$ time 
to compute $H_X$ and by Lemma~\ref{lem-mcb}, it takes $O(m\log^3m)$ time to
compute a MCB tree $\tilde{T}$ of $H_X$. In Step 2(c), it takes $O(m)$ time to 
compute a cylinder minor by Lemma~\ref{lem:minor}. From Property (3) of a 
good-separator, each edge of $H_X$ appears in at most two cycles which induce 
the subsets in $\AAA_X$. So Step 2(c) takes $O(m)$ time to compute $\AAA_X$. 
Therefore, the total time for Steps 2(a)-(c) is $O(m\log^3m)$. For distinct 
level 1 nodes $X$ and $X'$, the edge sets of subgraphs 
$(G|\overline{X})|{\cal Z}_X$ and $(G|\overline{X'})|{\cal Z}_{X'}$ are disjoint. 
From this, $\sum_{X:\mbox{level 1 node}} m_x=O(M)$. Therefore, the total time 
for Step is
\[
\sum_{X:\mbox{level 1 node}} O(m_x\log^3m_x)=O(M\log^3M).
\]
The time for other steps in Procedure Branch-Minor($G|U$) is $O(M)$. The number 
of recursive calls in which each vertex of $G|U$ is involved in the computation 
of Step 2 is $O(\frac{1}{\alpha})=O(1)$. Therefore, the running time of the 
algorithm is $O(n\log^3n)$.
\end{proof}

\section{$O(nk^2)$ time algorithm}
\label{sec-alg2}

To get an algorithm for Theorem~\ref{theo-3}, we follow the framework of
Procedure Branch-Minor in Section~\ref{sec-alg1} but use a different approach from
that in Steps 2(b)(c) to compute face separating cycles and a good separator for
${\cal Z}_X$ and $\overline{X}$. Our approach has the following major steps:
\begin{enumerate}
\item[(s1)] 
Given ${\cal Z}_X$ and $H_X$, for each $Z\in {\cal Z}_X$ the edges 
incident to face $f_Z$ in $H_X$ form a $(f_Z,f_{\overline X})$-separating cycle, 
denoted by $C_Z$ and called the {\em boundary cycle} of $Z$. Notice that the 
number of edges in $C_Z$, denoted by $|C_Z|$, is equal to $|\partial(Z)|$. 

For each $Z\in {\cal Z}_X$ with $|C_Z|\leq k$, we take $C_Z$ as a ``minimum''
$(f_Z,f_{\overline X})$-separating cycle and the cycle $C_Z$ induced subset $Z$
as a candidate for a noose induced edge subset $A_Z$ which separates $Z$ and
$\overline{X}$.

Notice that any two different boundary cycles share at most one common vertex, 
because otherwise it contradicts with that each $Z$ is a biconnected component.

\item[(s2)] 
Let ${\cal W}_X=\{Z\in {\cal Z}_X \mid |C_Z|>k$\}. 
We apply the techniques in \cite{KT12,KT13} to decompose $H_X$ into pieces 
(subgraphs), each piece contains face $f_Z$ for exactly one $Z\in {\cal W}_X$. 

\item[(s3)] 
For each piece containing one $f_Z$ with $Z\in\mathcal{W}_X$, we find a minimum
$(f_Z,f_{\overline{X}})$-separating cycle using the approaches in 
\cite{BSW13,Reif83}.

\item[(s4)] 
From the $(f_Z,f_{\overline X})$-separating cycles computed above, we
find non-crossing face separating cycles to get a a good separator for 
${\cal Z}_X$ and $\overline{X}$.
\end{enumerate}
The approach in \cite{Reif83} is a basic tool for Step (s3). The efficiency of the
tool can be improved by pre-computing some shortest distances between the vertices
in the vertex cut set separating the piece from the rest of the graph \cite{BSW13}. 
The pieces computed in Step (s2) have properties which allow us to use a scheme 
in \cite{KT12,KT13} to pre-compute some shortest distances (called $\upDDG$ and 
$\lowDDG$) for every piece to further improve the efficiency of the tool when it
is applied to the pieces. The separating cycles computed in Step (s3) may not be 
non-crossing because the unique shortest path assumption we used in 
Section~\ref{sec-alg1} does not hold in the scheme in \cite{KT12,KT13}. We develop 
a new technique to clear this hurdle. New ingredients in our approach also include: 
To find separating cycles, we use a simple method for $Z$ with $|C_Z|\leq k$ and 
use the complex techniques only for $Z$ with $|C_Z|>k$ instead of every 
$Z\in {\cal Z}_X$ as in \cite{KT12,KT13}. This reduces the time complexity by a 
$O(k)$ factor for finding the separating cycles. By the approaches of 
\cite{BSW13,Reif83} for finding the minimum face separating cycles, the scheme in 
\cite{KT12,KT13} for pre-computing $\upDDG$ and $\lowDDG$ and new developed 
technique to extract non-crossing separating cycles from the cycles computed in 
Steps (s1)-(s3), we find non-crossing separating cycles of length at most $k$ 
instead of $3k-1$ as in \cite{KT12,KT13}. 

\subsection{Review on previous techniques}
\label{sec-rev}

We now briefly review some notions and techniques introduced in \cite{KT12,KT13}. 
For a plane graph $G$, one face can be selected as the {\em outer-face}, denoted 
by $f_0$, and every face other than $f_0$ is called an {\em inner-face}. A plane 
graph is {\em almost triangulated} if every inner-face of the graph is incident 
to exactly three vertices and three edges. A plane graph is {\em $k$-outerplanar} 
if the normal distance from any vertex to $f_0$ is at most $k$. Let $\hat{G}$ be 
an almost triangulated graph. The height of vertex $u$ in $\hat{G}$ is 
$l_{\hat{G}}(u)=\nd_{\hat{G}}(V(f_0),u)$. The {\em height} of a
face $f$ of $\hat{G}$ is $l_{\hat{G}}(f)=\min_{u\in V(f)} l_{\hat{G}}(u)$.

A maximum connected set $Z$ of vertices of $\hat{G}$ is called a {\em crest} if
every vertex of $Z$ has the largest height in $\hat{G}$. For each $u$ with 
$l_{\hat{G}}(u)>0$, an arbitrary vertex $v$ adjacent to $u$ with 
$l_{\hat{G}}(v)<l_{\hat{G}}(u)$ (such a $v$ always exists) is selected as the 
{\em down vertex} of $u$ and the edge $\{u,v\}$ is called the {\em down edge} 
of $u$. When the down vertex of every vertex in $\hat{G}$ is selected, each 
vertex $u$ in $\hat{G}$ has a unique {\em down path} consisting of the selected 
down edges only.

For a path $L$ in $\hat{G}$, $d_{\hat{G}}(L)=\min_{u\in V(L)} l_{\hat{G}}(u)$ is
defined as the {\em depth} of $L$. A path $R$ between two crests $Z$ and $Z'$ is 
called a {\em ridge} between $Z$ and $Z'$ if $R$ has the maximum depth among all 
paths between $Z$ and $Z'$. A {\em crest separator} is a subgraph $S=L_1\cup L_2$ 
of $\hat{G}$, where $L_1$ is the unique down path of a vertex $u$ and $L_2$ is a 
path composed of the edge $\{u,u'\}$ and the unique down path of $u'$, $u'$ is not 
in $L_1$ and $l_{\hat{G}}(u')\leq l_{\hat{G}}(u)$. Vertices in $S$ of the largest 
height are called the {\em top vertices} and edge $\{u,u'\}$ is called the 
{\em top edge} of $S$. Note that each crest separator $S$ has $t\in\{1,2\}$ top 
vertices. The height $l_{\hat{G}}(S)$ of crest separator $S$ is the height of 
its top vertices. We say a crest separator $S=L_1\cup L_2$ is on a ridge $R$ if
a top vertex $u$ of $S$ is on $R$ and $l_{\hat{G}}(S)=d_{\hat{G}}(R)$. 
A crest separator $S=L_1\cup L_2$ is called {\em disjoint} if path $L_1$ 
and the down path of $u'$ do not have a common vertex, otherwise {\em converged}. 
For a converged crest separator $S$, the paths $L_1$ and $L_2$ have a common 
sub-path from a vertex other than $u$ to a vertex $w$ in $V(f_0)$. The vertex 
$v\neq u$ in the common sub-path with the largest height is called the 
{\em low-point} and the sub-path from $v$ to $w$ is called the 
{\em converged-path}, denoted by $\cp(S)$, of $S$. 

For $\hat{G}$ on the sphere $\Sigma$, let $\overline{f_0}$ be the region of
$\Sigma\setminus f_0$. A crest separator $S$ is a {\it crest separator for crests 
$Z$ and $Z'$} if (1) $S$ is on a ridge between $Z$ and $Z^\prime$ and (2) removing 
$S$ from $\overline{f_0}$ cuts $\overline{f_0}$ into two regions, one contains $Z$ 
and the other contains $Z'$. 
Given a set of $r-1$ crest separators $S_1,..,S_{r-1}$, removing $S_1,..,S_{r-1}$
from $\overline{f_0}$ cuts $\overline{f_0}$ into $r$ regions $R_1,..,R_r$. Let
$P_i, 1\leq i\leq r$, be the subgraph of $\hat{G}$ consisting of the edges of 
$\hat{G}$ in $R_i$ and the edges of every crest separator with its top edge 
incident to $R_i$. We call $P_i$ a piece ($P_i$ is called an extended component
in \cite{KT13,KT15}).
Let ${\cal W}$ be an arbitrary subset of $r$ crests $Z_1,..,Z_r$ in $\hat{G}$. 
It is known (implicitly in the proofs of Lemmas 6-8 of \cite{KT13} and
explicitly in Lemmas 3.6-3.9 of \cite{KT15}) that there is a set ${\cal S}$ 
of $r-1$ crest separators with the following properties:
\begin{itemize}
\item[(a)] The crest separators of ${\cal S}$ decompose $\hat{G}$ into pieces
$P_1,...,P_r$ such that each piece $P_i$ contains exactly one crest 
$Z_i\in {\cal W}$. Moreover, no crest separator in $\cal{S}$ contains a vertex of 
any $Z_i$ in $\cal{W}$.

\item[(b)] For each pair of pieces $P_i$ and $P_j$, there is a crest separator 
$S\in {\cal S}$ for $Z_i$ and $Z_j$ such that $S$ decomposes $\hat{G}$ into two 
pieces $P$ and $Q$, $P$ containing $Z_i$ and $Q$ containing $Z_j$. Moreover, $S$ 
has the minimum number of top vertices among the crest separators for $Z_i$ and 
$Z_j$.

\item[(c)]
Let $T_{\cal S}$ be the graph that $V(T_{\cal S})=\{P_1,...,P_r\}$ and there is 
an edge $\{P_i,P_j\}\in E(T_{\cal S})$ if there is a crest separator 
$S=E(P_i)\cap E(P_j)$ in ${\cal S}$. Then $T_{\cal S}$ is a tree.
\end{itemize}
The tuple $(\hat{G},{\cal S},{\cal W})$ is called a {\em good mountain structure 
tree} (GMST).
We call $T_{\cal S}$ the {\em underlying tree} of the GMST 
$(\hat{G},{\cal S},{\cal W})$.
For each edge $\{P_i,P_j\}$ in 
$T_{\cal S}$, the $S\in {\cal S}$ with $E(S)=E(P_i)\cap E(P_j)$ is called 
{\em the crest separator on edge $\{P_i,P_j\}$}. 
The following result is implied 
implicitly in \cite{KT13} and later stated in \cite{GX15} and \cite{KT15}.

\begin{lemma}\cite{KT13,KT15}
Given an arbitrary subset $\cal{W}$ of crests in an $O(k)$-outer planar $\hat{G}$, 
a GMST $(\hat{G},{\cal S},{\cal W})$ can be computed in $O(|V(\hat{G})|k)$ time.
\label{lem-gmst1}
\end{lemma}

Given a GMST $(\hat{G},{\cal S},{\cal W})$, we choose an arbitrary vertex $P_i$ in 
$T_{\cal S}$ as the root. Each crest separator $S\in {\cal S}$ decomposes 
$\hat{G}$ into two pieces, one contains $P_i$, called the {\em upper piece} by 
$S$, and the other does not, called the {\em lower piece} by $S$. A piece $P$ is 
{\em enclosed} by a converged crest separator $S$ if $P\setminus \cp(S)$ does 
not have any edge incident to $f_0$. For vertices $u$ and $v$ in a piece $P$, 
let $\dist_P(u,v)$ denote the length of a shortest path in $P$ between $u$ and 
$v$. For vertices $u$ and $v$ in a disjoint crest separator $S$, let 
$\dist_S(u,v)$ be the length of the path in $S$ between $u$ and $v$. For vertices
$u$ and $v$ in a converged crest separator $S$, let $\dist_S(u,v)$ be the length 
of the shortest path in $S$ between $u$ and $v$ if at least one of $u$ and $v$ is
in $\cp(S)$, otherwise let $\dist_S(u,v)$ be the length of the path in $S$ between 
$u$ and $v$ that does not contain the the low-point of $S$.

A disjoint crest separator $S$ decomposes $\hat{G}$ into two pieces $P$ and $Q$.
For $P$ (resp. $Q$), let $G_{SP}$ (resp. $G_{SQ}$) be the weighted graph on the 
vertices in $S$ such that for every pair of vertices $u$ and $v$ in $S$, if 
$\dist_P(u,v)<\dist_S(u,v)$ (resp. $\dist_Q(u,v)<\dist_S(u,v)$) then there is an 
edge $\{u,v\}$ with weight $\dist_P(u,v)$ in $G_{SP}$ (resp. with weight 
$\dist_Q(u,v)$ in $G_{SQ}$). If $P$ is the upper piece by $S$ then $G_{SP}$ is 
called the $\upDDG(S)$ and $G_{SQ}$ the $\lowDDG(S)$, otherwise $G_{SP}$ is called 
the $\lowDDG(S)$ and $G_{SQ}$ the $\upDDG(S)$.

A converged crest separator $S$ decomposes $\hat{G}$ into two pieces and exactly 
one piece $P$ is enclosed by $S$. For $P$, let $G_{SP}$ be defined as in the
previous paragraph. Let $Q$ be the other piece not enclosed by $S$. A plane graph 
$Q'$ can be created from $Q$ by cutting $Q$ along $\cp(S)$: create a duplicate 
$v'$ for each vertex $v$ in $\cp(S)$ and create a duplicate $e'$ for each edge 
$e$ in $\cp(S)$ (see Figure~\ref{fig-pcs}). Let $S'$ be the subgraph induced by 
the edges of $S$ and the duplicated edges. For every pair of vertices $u,v$ in 
$S'$, let $\dist_{S'}(u,v)$ be the length of the path in $S'$ between $u$ and $v$. 
For $Q'$, let $G_{SQ'}$ be the weighted graph on the vertices on $S'$ such that 
for every pair of vertices $u$ and $v$, if $\dist_{Q'}(u,v)<\dist_{S'}(u,v)$ then 
there is an edge $\{u,v\}$ with weight $\dist_{Q'}(u,v)$ in $G_{SQ'}$. If $P$ is 
the upper piece by $S$ then $G_{SP}$ is called the $\upDDG(S)$ and $G_{SQ'}$ the 
$\lowDDG(S)$, otherwise $G_{SP}$ is called the $\lowDDG(S)$ and $G_{SQ'}$ the 
$\upDDG(S)$.

\begin{figure}[t]
\centerline{\includegraphics[scale=0.9]{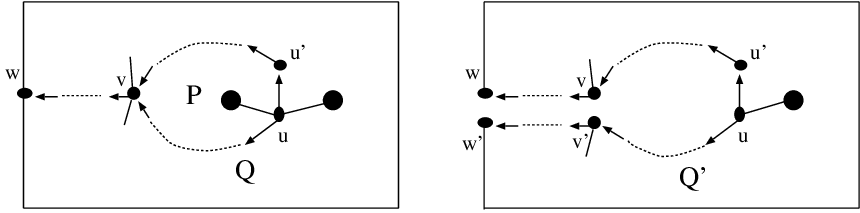}}
\caption{Graph $Q'$ obtained from cutting $Q$ along $\cp(S)$.} 
\label{fig-pcs}
\end{figure}

Each crest separator $S\in {\cal S}$ decomposes $\hat{G}$ into two pieces $P$ and 
$Q$ and we assume $P$ is the upper piece and $Q$ is the lower piece. For each edge 
$e=\{u,v\}$ in $\upDDG(S)$ (resp. $\lowDDG(S)$), the weight of $e$ is used to 
decide whether a minimum face separating cycle should use a shortest path between 
$u$ and $v$ in $P$ (resp. $Q$) or not, and if so, any shortest path shortest path 
between $u$ and $v$ in $P$ (resp. $Q$) can be used (there may be multiple shortest 
paths between $u$ and $v$ in $P$ (resp. $Q$)). So we say edge $\{u,v\}$ in 
$\upDDG(S)$ (resp. $\lowDDG(S)$) {\em represents} any shortest path between $u$ 
and $v$ in $P$ (resp. $Q$). The computation of $\upDDG(S)$ (resp.$\lowDDG(S)$) 
also includes computing one shortest path between $u$ and $v$ in $P$ (resp. $Q$) 
for every edge $\{u,v\}$ in $\upDDG(S)$ (resp. $\lowDDG(S)$).

Note that in this paper we use the terms $\upDDG$ and $\lowDDG$ instead of the 
{\em h-high pseudo shortcut set} in \cite{KT13} to give a more clear description.

The following properties of the down paths, GMST $(\hat{G},{\cal S},{\cal W})$, 
$\upDDG(S)$ and $\lowDDG(S)$ can be easily verified and are proved in \cite{KT13}: 
\begin{itemize}
\item[(I)] 
For any pair of vertices $u$ and $v$ in $\hat{G}$, 
$\dist_{\hat{G}}(u,v)\geq|l_{\hat{G}}(u)-l_{\hat{G}}(v)|$.
\item[(II)]For any pair of vertices $u$ and $v$ in a same down path of $S$, 
$\dist_{S}(u,v)=|l_{\hat{G}}(u)-l_{\hat{G}}(v)|=\dist_{\hat{G}}(u,v)$.
\item[(III)] If $\hat{G}$ is $O(k)$-outerplanar then for every $S\in {\cal S}$, 
there are $O(k)$ vertices and $O(k)$ edges in $S$ and every edge in $\upDDG(S)$ 
($\lowDDG(S))$ has weight $O(k)$.
\item[(IV)] If $\hat{G}$ is $O(k)$-outerplanar, 
$\sum_{P_i\in T_{\cal{S}}}|E(P_i)|=|E(\hat{G})|+O(|{\cal{S}}|k)$.
\item[(V)]
For every $S\in {\cal S}$ and each edge $e=\{u,v\}$ in $\upDDG(S)$/$\lowDDG(S)$, 
any shortest path represented by $e$ contains no vertex of height greater than 
$l_{\hat{G}}(S)$ and no more than $t-1$ vertices of height $l_{\hat{G}}(S)$, 
where $t$ is the number of top vertices of $S$. 
\item[(VI)]
Let $S$ be the crest separator on edge $\{P_i,P_j\}$ in $T_{\cal S}$. Assume that 
$Z_i$ and $Z_j$ are in the upper piece and lower piece by $S$, respectively. 
For every edge $e=\{u,v\}$ in $\upDDG(S)$ (resp. $\lowDDG(S)$), any shortest path 
represented by $e$ and the segment of $S$ between $u$ and $v$ that contains a top 
vertex of $S$ form a cycle which separates $Z_i$ (resp. $Z_j$) from $f_0$.
\end{itemize}
From Properties (I)-(VI), the $\upDDG(S)$ and $\lowDDG(S)$ can be computed 
as shown in the next lemma (Lemma 18 in \cite{KT13}).
\begin{lemma} \cite{KT13} 
Given a GMST $(\hat{G},{\cal S},{\cal W})$, $\upDDG(S)$ and $\lowDDG(S)$ for 
all $S\in {\cal S}$ can be computed in $O(|V(\hat{G})|k^3)$ time. 
\label{lem-ddg}
\end{lemma}

The authors of \cite{KT13} settle for this result because in their application, 
$|{\cal W}|=O(|V(\hat{G})|)$. However, it is hidden in the proof details and 
stated in \cite{GX15} that the time complexity is actually 
$O(|V(\hat{G})|k+|{\cal W}|k^3)$. It is also hidden in the details that the 
Lemma holds when $\hat{G}$ is weighted. We now state Lemma \ref{lem-ddg} as the 
following Lemma which will be used in this paper.
\begin{lemma} \cite{KT13,KT15}
Given a GMST $(\hat{G},{\cal S},{\cal W})$, $\upDDG(S)$ and $\lowDDG(S)$ for all 
$S\in {\cal S}$ can be computed in $O(|V(\hat{G})|k+|{\cal W}|k^3)$ time. 
\label{lem-ddg1}
\end{lemma}

\subsection{Algorithm for Theorem~\ref{theo-3}}
\label{sec-newalg}

For a level 1 node $X$ and the set ${\cal Z}_X$ of child nodes in the layer tree 
$\LT(G|U,\partial(U))$ in Procedure Branch-Minor, recall that $G_X$ is the plane
graph converted from plane hypergraph $(G|\overline{X})|{\cal Z}_X$ and $H_X$ is 
the weighted plane graph computed from $G_X$ as described in Section~\ref{sec-alg1}. 
Recall that ${\cal W}_X=\{Z\in {\cal Z}_X\mid |C_Z|>k\}$. We apply the techniques in 
\cite{KT12,KT13} to decompose $H_X$ into pieces, each piece contains face $f_Z$ 
of $H_X$ for exactly one $Z\in {\cal W}_X$. It may not be straightforward to 
decompose $H_X$ directly by the techniques of \cite{KT12,KT13} because some of 
the techniques are described for graphs while $H_X$ is weighted (edges have weight 
$1/2$ or 1). To get a decomposition of $H_X$ as required, we first construct an 
almost triangulated graph $\hat{G}_X$ from $G_X$ with each $Z\in {\cal Z}_X$ 
represented by a crest of $\hat{G}_X$; then by the techniques of \cite{KT12,KT13}
find a GMST of $\hat{G}_X$ which decomposes $\hat{G}_X$ into pieces, each piece
contains exactly one crest; next construct an almost triangulated weighted graph 
$\hat{H}_X$ from $H_X$ with each $Z\in {\cal Z}_X$ represented by a crest of 
$\hat{H}_X$; and finally compute a set of crest separators in $\hat{H}_X$ based 
on the GMST of $\hat{G}_X$ to decompose $\hat{H}_X$ into pieces such that each 
piece $\hat{H}_X$ contains exactly one crest $Z\in {\cal W}_X$ (and thus each 
piece of $H_X$ contains exactly one face $f_Z$).

We first describe the construction of $\hat{G}_X$. Let $f_{\overline{X}}$ be the
outer face of $G_X$. For every $Z\in {\cal Z}_X$, we add a vertex, also denoted 
by $Z$, and edges $\{u,Z\}$ for every $u\in V(f_Z)$ to face $f_Z$ in $G_X$. For 
every natural face $f$ of $G_X$ with $|V(f)|>3$, we select an arbitrary vertex 
$v$ of $V(f)$ with $l_{G_X}(v)=l_{G_X}(f)$ as the {\em low-point} of $f$, denoted
by $\lp(f)$, and we add edges $\{u,\lp(f)\}$ to face $f$ for every $u\in V(f)$ and 
not adjacent to $\lp(f)$. Let $\hat{G}_X$ be the graph obtained from adding the 
vertices and edges above. Let $f_{\overline{X}}$ be the outer face of $\hat{G}_X$. 
Then $\hat{G}_X$ is almost triangulated. 

$G_X$ is a subgraph of $\hat{G}_X$. For every $u\in V(G_X)\cap V(\hat{G}_X)$, 
$l_{G_X}(u)=l_{\hat{G}_X}(u)$, every vertex $Z$ added to face $f_Z$ of $G_X$ is 
a crest of $\hat{G}_X$ and every crest of $\hat{G}_X$ is a vertex $Z$ added to
$f_Z$. Recall that ${\cal W}_X=\{Z\in {\cal Z}_X \mid |C_Z|>k\}$ which is a subset 
of crests in $\hat{G}_X$. By Lemma~\ref{lem-gmst1}, we can find a GMST 
$(\hat{G}_X,{\cal S},{\cal W}_X)$. 

Next we describe how to construct $\hat{H}_X$. For every $Z\in {\cal Z}_X$, we add 
a vertex, also denoted by $Z$, and edges $\{u,Z\}$ for every $u\in V(f_Z)$ to face
$f_Z$ of $H_X$. We assign each edge $\{u,Z\}$  weight $1$. Let $\hat{H}_X$ be the
graph computed above and $f_{\overline{X}}$ be the outer face of $\hat{H}_X$. Then 
$\hat{H}_X$ is almost triangulated. Notice that 
$V(G_X)\subseteq V(\hat{G}_X)\subseteq V(\hat{H}_X)$, 
$V(H_X)\subseteq V(\hat{H}_X)$, $E(G_X)\subseteq E(\hat{G}_X)$, 
$E(G_X)\subseteq E(H_X)\subseteq E(\hat{H}_X)$ and $|V(\hat{H}_X)|=O(|V(G_X)|)$. 
We define the height $l_{\hat{H}_X}(u)$ of each vertex $u$ of $\hat{H}_X$ as 
follows:
\begin{itemize}
\item $l_{\hat{H}_X}(u)=l_{G_X}(u)$ if $u\in V(\hat{H}_X)\cap V(G_X)$. 
\item $l_{\hat{H}_X}(u)=l_{G_X}(f)+1/2$ if $u=u_f$ is the vertex added to a 
natural face $f$ of $G_X$. 
\item  $l_{\hat{H}_X}(u)=l_{G_X}(f_Z)+1$ if $u=Z$ is the vertex added to $f_Z$. 
\end{itemize}
Then each vertex $Z$ is a crest of $\hat{H}_X$ and each crest of $\hat{H}_X$ is 
a vertex $Z$. The heights of a face and a crest separator, and the depth of
a ridge in $\hat{H}_X$ are defined based on $l_{\hat{H}_X}(u)$ similarly as those 
in Section~\ref{sec-rev}.

Similar to the down vertex and down edge in $\hat{G}_X$, we define the down vertex
and down edge for each vertex of $\hat{H}_X$. Recall that any vertex $v$ adjacent
to vertex $u$ with $l_{\hat{H}_X}(v)<l_{\hat{H}_X}(u)$ can be selected as the down 
vertex of $u$. We choose the down vertex for each $u$ of $\hat{H}_X$ as follows:
\begin{itemize}
\item if $u\in V(\hat{G}_X)$ and the down edge $\{u,v\}$ in $\hat{G}_X$ is an
edge of $G_X$ or an edge added to face $f_Z$ ($u$ is a crest) then $v$ is the 
down vertex of $u$ in $\hat{H}_X$;
\item if $u\in V(\hat{G}_X)$ and the down edge $\{u,v\}$ in $\hat{G}_X$ is an
edge added to a natural face $f$ of $G_X$ then the vertex $u_f$ added to $f$ in 
$\hat{H}_X$ is the down vertex of $u$; and
\item otherwise, $u$ is not in $\hat{G}_X$ and is the vertex $u_f$ added to a 
natural face $f$ of $G_X$ in $\hat{H}_X$; then the low-point $\lp(f)$ is the down 
vertex of $u$.
\end{itemize}
The edge between vertex $u$ and its down vertex is the down edge of $u$.

Given a GMST $(\hat{G}_X,{\cal S},{\cal W}_X)$, for every crest separator 
$S\in {\cal S}$ and every edge $e$ of $S$, either $e\in E(G_X)\cap E(\hat{G}_X)$ 
or $e$ is an edge added to a natural face $f$ of $G_X$ during the construction of 
$\hat{G}_X$. We convert each $S\in {\cal S}$ into a subgraph $D$ of $\hat{H}_X$: 
for every edge $e=\{u,v\}$ in $S$, $\{u,v\}$ of $\hat{H}_X$ is included in $D$ if 
$e\in E(G_X)\cap E(\hat{G}_X)$, otherwise edges $\{u,u_f\},\{u_f,v\}$ of 
$\hat{H}_X$ are included in $D$, where $u_f$ is the vertex added to face $f$ of 
$G_X$ when $H_X$ is created from $G_X$. 

A crest separator $S\in {\cal S}$ for crests $Z$ and $Z'$ consists  of two paths 
$L_1$ and $L_2$ ($L_1$ is the down path of some vertex $u$ and $L_2$ is composed 
of the top edge $\{u,u'\}$ and the down path of $u'$, where $u'$ is not in $L_1$ 
and $l_{\hat{G}_X}(u')\leq l_{\hat{G}_X}(u)$) and decomposes $\hat{G}_X$ into two 
pieces, one contains $Z$ and the other contains $Z'$. From the way we define the 
down vertex of every vertex $u$ in $\hat{H}_X$, the subgraph $D$ converted from 
$S$ is also a crest separator consisting of two paths $L'_1$ and $L'_2$ in 
$\hat{H}_X$ described below:
\begin{itemize}
\item[(1)] if the top edge $\{u,u'\}$ of $S$ is an edge of $G_X$ then $L'_1$ is the 
down path from vertex $u$ and $L'_2$ is composed of the top edge $\{u,u'\}$ and 
the down path from $u'$; 

\item[(2)] if $\{u,u'\}$ is an added edge to a natural face $f$ of $G_X$ when 
constructing $\hat{G}_X$ and $l_{\hat{G}_X}(u)=l_{\hat{G}_X}(u')+1$ then $L'_1$ is 
the down path from vertex $u$ and $L'_2$ is composed of the top edge $\{u,u_f\}$ 
and the down path from $u_f$, where $u_f$ is the vertex added to $f$ when 
constructing $H_X$ and $l_{\hat{H}_X}(u_f)=l_{\hat{H}_X}(u)-1/2$;
\item[(3)] otherwise ($\{u,u'\}$ is an added edge to $f$ and 
$l_{\hat{G}_X}(u)=l_{\hat{G}_X}(u')$), $L'_1$ is the down path from vertex $u_f$ 
(added to $f$ when constructing $H_X$) and $L'_2$ is composed of the top edge 
$\{u_f,\lp(f)\}$ and the down path from $\lp(f)$, where $\lp(f)=u$ or $\lp(f)=u'$.
\end{itemize}
In Case (1) and Case (2), $D$ is a crest separator for $Z$ and $Z'$. In Case (3), 
either $D$ intersects a ridge $R$ between $Z$ and $Z'$ but 
$l_{\hat{H}_X}(u_f)=d_{\hat{H}_X}(R)+1/2=l_{\hat{H}_X}(u)+1/2$ or $D$ is on a 
ridge $R$ between $Z$ and $Z'$ and 
$l_{\hat{H}_X}(D)=l_{\hat{H}_X}(u_f)=d_{\hat{H}_X}(R)$. In all cases, $D$ is a 
crest separator which decomposes $\hat{H}_X$ into two pieces $P$ and $Q$, $P$ 
contains $Z$ and $Q$ contains $Z'$.

Let ${\cal D}=\{D|D \mbox{ is converted from $S$ for every $S\in {\cal S}$}\}$ and 
assume that ${\cal W}_X$ has $r$ crests $Z_1,..,Z_r$. It is easy to see that
${\cal D}$ has the following properties:
\begin{itemize}
\item[(A)] The crest separators of ${\cal D}$ decompose $\hat{H}_X$ into $r$ 
pieces $P_1,...,P_r$ such that each piece $P_i$ has exactly one crest $Z_i$.
Moreover, no crest separator $D$ in ${\cal D}$ contains a crest in ${\cal W}_X$, 
that is, each piece contains the edges in $E(f_Z)$ of $H_X$ for exactly one 
$Z\in {\cal W}_X$.
\item[(B)] For each pair of pieces $P_i$ and $P_j$, there is a crest separator 
$D\in {\cal D}$ such that $D$ decomposes $\hat{H}_X$ into two pieces $P$ and $Q$, 
$P$ contains $Z_i$ and $Q$ contains $Z_j$. 

Let ${\cal D}_{ij}$ be the set of crest separators for $Z_i$ and $Z_j$ in 
$\hat{H}_X$. Recall that every crest separator for $Z_i$ and $Z_j$ is on a ridge 
between $Z_i$ and $Z_j$ and thus all crest separators in ${\cal D}_{ij}$ have the 
same height, denoted as $l_{\hat{H}_X}({\cal D}_{ij})$. For $D$ in Case (1), Case
(2) and Case (3) with 
$l_{\hat{H}_X}(u_f)=d_{\hat{H}_X}(R)=l_{\hat{H}_X}({\cal D}_{ij})$, from Property
(B), we have

(B1) $D$ is in ${\cal D}_{ij}$ and has the minimum number of top vertices among 
all crest separators in ${\cal D}_{ij}$.

For $D$ in Case (3) with $l_{\hat{H}_X}(u_f)=d_{\hat{H}_X}(R)+1/2$,

(B2) $D$ is not in ${\cal D}_{ij}$, 
$l_{\hat{H}_X}(D)=l_{\hat{H}_X}({\cal D}_{ij})+1/2$ and every crest separator in
${\cal D}_{ij}$ has two top vertices.
\item[(C)]
Let $T_{\cal D}$ be the graph that $V(T_{\cal D})=\{P_1,...,P_r\}$ and 
there is an edge $\{P_i,P_j\}\in E(T_{\cal D})$ if there is a crest separator
$D=E(P_i)\cap E(P_j)$ in ${\cal D}$. Then $T_{\cal D}$ is a tree.
\end{itemize}
The tuple $(\hat{H}_X,{\cal D},{\cal W}_X)$ is called a {\em pseudo good mountain 
structure tree (pseudo GMST)} for ${\cal W}_X$. For each $D\in {\cal D}$ and
vertices $u,v$ of $D$, $\dist_D(u,v)$, $\upDDG(D)$ and $\lowDDG(D)$ are defined 
similarly as $\dist_S(u,v)$, $\upDDG(S)$ and $\lowDDG(S)$ for crest separator $S$ 
and $u,v$ of $S$ in Section~\ref{sec-rev}. 
As shown in the next lemma, Properties (I)-(VI) in Section~\ref{sec-rev} hold 
for a pseudo GMST, $\upDDG(D)$ and $\lowDDG(D)$. 
\begin{lemma}
Properties (I)-(VI) hold for a pseudo GMST $(\hat{H}_X,{\cal D},{\cal W}_X)$, 
$\upDDG(D)$ and $\lowDDG(D)$, $D\in {\cal D}$.
\label{lem-gmst2}
\end{lemma}
\begin{proof}
Property (I) and Property (II) hold trivially from the definition of 
$l_{\hat{H}_X}(u)$ for vertex $u$ and the definition of down path.

If $\hat{H}_X$ is $O(k)$-outerplanar, then the height of $D$ is $O(k)$ and there 
are $O(k)$ vertices and $O(k)$ edges in $D$ because the weight of each edge in 
$\hat{H}_X$ is $1$ or $1/2$. From Property (II), every edge in $\upDDG(D)$ 
($\lowDDG(D))$ has weight $O(k)$ because $\dist_{D}(u,v)=O(k)$ for any $u$ and 
$v$ in $D$. Thus, Property (III) holds.

Every $D\in {\cal D}$ is converted from an $S\in {\cal S}$ and $|E(D)|=O(|E(S)|)$.
From this, and $|{\cal D}|=|{\cal S}|$, 
$\sum_{P_i\in T_{\cal{D}}}|E(P_i)|=|E(\hat{H}_X)|+O(|{\cal{D}}|k)$.
Thus, Property (IV) holds.

From Property (II), for each edge $e=\{u,v\}$ in $\upDDG(S)$/$\lowDDG(S)$, $u$ and 
$v$ must be in different down paths of $D$ and $\dist_{D}(u,v)=2l_{\hat{H}_X}(D)-
l_{\hat{H}_X}(u)-l_{\hat{H}_X}(v)+(t-1)$, where $t\in\{1,2\}$ is the number of top 
vertices of $D$. Any vertex $w$ of height greater than $l_{\hat{H}_X}(D)$ in 
$\hat{H}_X$ has height at least $l_{\hat{H}_X}(D)+1/2$. From Property (I), any path 
between $u$ and $v$ that contains $w$ has length at least 
$2l_{\hat{H}_X}(D)+1-l_{\hat{H}_X}(u)-l_{\hat{H}_X}(v)\geq \dist_D(u,v)$. 
Similarly, any path between $u$ and $v$ that contains $t$ vertices of height 
$l_{\hat{H}_X}(D)$ has length at least $\dist_D(u,v)$. Therefore, any path 
represented by $e$ contains no vertex of height greater than $l_{\hat{H}_X}(D)$ 
and no more than $t-1$ vertices of height $l_{\hat{H}_X}(D)$ since its length is 
strictly smaller than $\dist_D(u,v)$. Thus, Property (V) holds.

We prove Property (VI) by contradiction. Let $D\in {\cal D}$ be the crest 
separator on edge $\{P_i,P_j\}$ in $T_{\cal D}$. We assume that $Z_i$ and $Z_j$ 
are in the upper piece and lower piece by $D$ respectively. For each edge
$e=\{u,v\}$ in $\upDDG(D)$ (resp. $\lowDDG(D)$), let $P_e$ be any shortest path
in the upper piece (resp. lower piece) represented by $e$ and let $D(u,v)$ be
the segment of $D$ between $u$ and $v$ and containing a top vertex of $D$. Then
$P_e$ and $D(u,v)$ form a cycle in $\hat{H}_X$. Assume for contradiction that the 
cycle formed by $P_e$ and $D(u,v)$ does not separate $Z_i$ (resp. $Z_j$) from 
$f_{\overline{X}}$. Let ${\cal D}_{ij}$ be the set of crest separators for $Z_i$
and $Z_j$ as defined in Property (B). Let $D'$ be the subgraph of $\hat{H}_X$ 
obtained by replacing $D(u,v)$ with $P_e$ in $D$. Then $D'$ separates $Z_i$ from 
$Z_j$ and thus, intersects with every ridge between $Z_i$ and $Z_j$. Since each 
ridge between $Z_i$ and $Z_j$ has depth $l_{\hat{H}_X}({\cal D}_{ij})$, $D'$ 
contains at least one vertex $w$ with 
$l_{\hat{H}_X}(w)=l_{\hat{H}_X}({\cal D}_{ij})$. Notice that no vertex of 
$D'\setminus P_e$ has height $l_{\hat{H}_X}({\cal D}_{ij})$. So $P_e$ contains 
$w$. Consider Case (B1) in Property (B) where $D\in {\cal D}_{ij}$ and $D$ has the 
minimum number of top vertices among all crest separators in ${\cal D}_{ij}$.
If $D$ has exactly one top vertex, from Property (V), $P_e$ contains no vertex of 
height $l_{\hat{H}_X}({\cal D}_{ij})$, a contradiction. If $D$ has two top vertices 
then every crest separator in ${\cal D}_{ij}$ (for $Z$ and $Z'$) has two top 
vertices and $w$ is the only vertex of height $l_{\hat{H}_X}({\cal D}_{ij})$ in 
$P_e$. However, this means that there is a crest separator for $Z_i$ and $Z_j$ 
with exactly one top vertex $w$, a contradiction. Consider Case (B2) where 
$l_{\hat{H}_X}(D)=l_{\hat{H}_X}({\cal D}_{ij})+1/2$, and every crest separator in
${\cal D}_{ij}$ has two top vertices. As proved, if $P_e$ has exactly one vertex 
of height $l_{\hat{H}_X}({\cal D}_{ij})$, we have a contradiction. If $P_e$ has two 
vertices of height greater than or equal to $l_{\hat{H}_X}({\cal D}_{ij})$, the 
length of $P_e$ is not strictly smaller than the length of $\dist_{D}(u,v)$, a 
contradiction. This gives Property (VI).
\end{proof}

Given a GMST $(\hat{G}_X,{\cal S},{\cal W}_X)$, a pseudo GMST 
$(\hat{H}_X,{\cal D},{\cal W}_X)$ can be computed in $O(|{\cal S}|k)$ time.
Lemma~\ref{lem-ddg} and Lemma~\ref{lem-ddg1} hold for a pseudo GMST because the 
lemmas only rely on Properties (I)-(VI). The computation of $\upDDG(D)$ and 
$\lowDDG(D)$ in Lemma~\ref{lem-ddg} and Lemma~\ref{lem-ddg1} 
(implicitly) uses a linear time algorithm by Thorup \cite{T99} for the single 
shortest path problem in graphs with integer edge weight as a subroutine. This 
requires that the edge weight in $\hat{H}_X$, $\upDDG(D)$ and $\lowDDG(D)$ can be 
expressed in $O(1)$ words in integer form. So the perturbation technique used in 
Section~\ref{sec-alg1} cannot be used in the computation above. In the algorithm 
for Theorem~\ref{theo-3}, we do not assume the uniqueness of shortest paths when 
we compute minimum separating cycles. A minimum $(f_Z,f_{\overline X})$-separating 
cycle $C$ cuts $\Sigma$ into two regions, one contains $f_Z$ and the other 
contains $f_{\overline X}$. Let $\ins(C)$ be the region containing $f_Z$. We say 
two cycles $C$ and $C'$ {\em cross} with each other if 
$\ins(C)\cap \ins(C')\neq \emptyset$, $\ins(C)\setminus \ins(C')\neq \emptyset$ 
and $\ins(C')\setminus \ins(C)\neq \emptyset$.
We say a set of at least two cycles are {\em crossing} if for any cycle $C$ in 
this set, there is at least one other cycle $C'$ that crosses with $C$. In 
Section~\ref{sec-alg1}, we do not have crossing minimum separating cycles because 
of the uniqueness of shortest paths. However, two minimum separating cycles 
computed without the assumption of unique shortest paths may cross with each other. 
We eliminate each crossing cycle set by exploiting some new properties of pseudo
GMST, $\upDDG(D)$ and $\lowDDG(D)$, and then compute a good separator for 
${\cal A}_X$.

Recall that $\hat{H}_X$ is constructed from $H_X$ by adding vertex (crest) $Z$ 
and edges $\{u,Z\}$, $u\in V(f_Z)$, to face $f_Z$ in $H_X$. Each piece $P$ of a 
pseudo GMST $(\hat{H}_X,{\cal D},{\cal W}_X)$ contains exactly one crest 
$Z\in {\cal W}_X$. By Property (A), removing $Z$ and the edges incident to $Z$ 
from piece $P$ gives a piece of $H_X$ containing face $f_Z$ for exactly one 
$Z\in {\cal W}_X$. From this and further Property (A) and Property (V), we will 
show in the proof of Lemma~\ref{lem-all-cycle} that a 
minimum $(f_Z,f_{\overline{X}})$-separating cycle can be computed based on
$(\hat{H}_X,{\cal D},{\cal W}_X)$, $\upDDG(D)/\lowDDG(D)$, $D\in {\cal D}$, and
the approaches of \cite{BSW13,Reif83}.
 
Given a set ${\cal W}_X$ of crests in $\hat{H}_X$, our algorithm for 
Theorem~\ref{theo-3} computes a pseudo GMST $(\hat{H}_X,{\cal D},{\cal W}_X)$, 
calculates $\upDDG(D)$ and $\lowDDG(D)$ for every crest separator $D\in {\cal D}$, 
and finds a minimum $(f_Z,f_{\overline X})$-separating cycle for every 
$Z\in {\cal W}_X$ using the pseudo GMST, $\upDDG(D)$ and $\lowDDG(D)$. 
More specifically, we replace Steps 2(b)(c) in Procedure Branch-Minor with the 
following subroutine to get an algorithm for Theorem~\ref{theo-3}.\\
{\bf Subroutine} Crest-Separator\\
{\bf Input:} Hypergraph $(G|\overline{X})|{\cal Z}_X$, weighted graph $H_X$ and 
integer $k$.\\
{\bf Output:} Good separator $\AAA_X$ for ${\cal Z}_X$ and $\overline{X}$
or a $(k+1)\times \ceil{\frac{k+1}{2}}$ cylinder minor.\\
\begin{enumerate}
\item[(1)] For every $Z\in {\cal Z}_X$ with $|C_Z|\leq k$, take $C_Z$ as
a $(f_Z,f_{\overline X})$-separating cycle and mark $Z$ as separated. 
\item[(2)] 
Compute $\hat{G}_X$ and $\hat{H}_X$. 
Let ${\cal W}_X=\{Z\in {\cal Z}_X \mid |C_Z|>k\}$.

\item[(3)] Compute a GMST $(\hat{G}_X,{\cal S},{\cal W}_X)$ by 
Lemma~\ref{lem-gmst1} and a pseudo GMST $(\hat{H}_X,{\cal D},{\cal W}_X)$ from
$(\hat{G}_X,{\cal S},{\cal W}_X)$.

\item[(4)] Compute $\upDDG(D)$ and $\lowDDG(D)$ for every crest separator
$D\in {\cal D}$ by Lemma~\ref{lem-ddg1}.

\item[(5)] Mark every $Z\in {\cal W}_X$ as un-separated. Repeat the following 
until every $Z$ is marked as separated.

\begin{enumerate}
\item[(5.1)] 
Choose an arbitrary un-separated $Z$, compute a 
minimum $(f_Z,f_{\overline X})$-separating cycle $C$ using the pseudo GMST
$(\hat{H}_X,{\cal D},{\cal W})$, $\upDDG(D)$ and $\lowDDG(D)$. We call $C$ the
{\em cycle computed for } $Z$.

\item[(5.2)] If the length of $C$ is greater than $k$ then compute a
$(k+1)\times \ceil{\frac{k+1}{2}}$ cylinder minor by Lemma~\ref{lem:minor}
and terminate. Otherwise, take this cycle as the minimum
$(f_Z,f_{\overline X})$-separating cycle for every $Z\in \ins(C)$ and mark every
$Z$ in $\ins(C)$ separated.
\end{enumerate}

\item[(6)] Compute $\AAA_X$ from the (minimum) face separating cycles obtained
in Step (1) and Step (5).
\end{enumerate}

We now analyze the time complexity of Subroutine Crest-Separator.

\begin{lemma}
Steps (1)-(5) of Subroutine Crest-Separator can be computed in $O(|E(H_X)|k^2)$ 
time.
\label{lem-all-cycle}
\end{lemma}
\begin{proof}
Let $m$ be the numbers of edges in $H_X$. Then $|E(\hat{H}_X)|=O(m)$. Notice that 
each edge of $\hat{H}_X$ appears in at most one boundary cycle $C_Z$ because any
pair of boundary cycles do not have more than one common vertex. From this, 
$\sum_{Z\in {\mathcal{Z}_X}}|C_Z|=O(m)$. Therefore, for all 
$Z\in {\cal Z}_X\setminus{\cal W}_X$, $(f_Z,f_{\overline X})$-separating cycles 
can be found in $O(m)$ time. For each crest $Z\in {\cal W}_X$, $C_Z$ has more 
than $k$ edges in $H_X$. Therefore, $|{\cal W}_X|=O(m/k)$. A pseudo GMST 
$(\hat{H}_X,{\cal D},{\cal W}_X)$, $\upDDG(D)$ and $\lowDDG(D)$ for all 
$D\in\cal{D}$ can be computed in $O(mk+|{\cal W}_X|k^3)=O(mk^2)$ time 
(Lemmas \ref{lem-gmst1} and \ref{lem-ddg1}).

For an un-separated crest $Z\in {\cal W}_X$, let $P$ be the piece in the pseudo
GMST $(\hat{H}_X,{\cal D},{\cal W}_X)$ containing $Z$ 
(see Figure~\ref{fig-findcycle} (a)). By Property (A), removing $Z$ and the edges 
incident to $Z$ from $P$ gives a piece $P'$ of $H_X$ containing face $f_Z$ of 
$H_X$ (see Figure~\ref{fig-findcycle} (b)). Let $x$ be an arbitrary vertex in $P$ 
incident to $Z$ and $L$ be the down path from $x$ to a vertex 
$w\in V(f_{\overline{X}})$. Then $L$ is a shortest path between $x$ and $w$ in 
$H_X$ and can be found in $O(k)$ time. Let $P(L)$ be the weighted plane graph 
obtained from $P'$ by cutting along $L$: for each vertex $u$ in $L$ create a 
duplicate $u'$, for each edge $e$ in $L$ create a duplicate $e'$ and create a 
new face bounded by edges of $E(f_{\overline{X}})$, $E(f_Z)$, $L$ and their 
duplicates (see Figure~\ref{fig-findcycle} (c)). Let $H_X(L)$ be the weighted 
graph obtained from $H_X$ by replacing $P'$ with $P(L)$. For each vertex $u$ in 
path $L$, let $C_u$ be a shortest path between $u$ and its duplicate $u'$ in 
$H_X(L)$. Let $y$ be a vertex in $L$ such that $C_y$ has the minimum length among 
the paths $C_u$ for all vertices $u$ in $L$. Then $C_y$ gives a minimum 
$(f_Z,f_{\overline X})$-separating cycle in $H_X$ \cite{Reif83}. 

\begin{figure}[t]
\centerline{\includegraphics[scale=0.9]{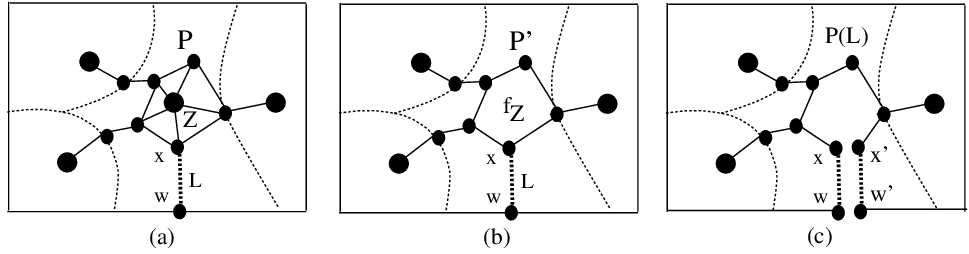}}
\caption{(a) Piece $P$ of $\hat{H}_X$, (b) piece $P'$ of $H_X$ and (c) graph 
$P(L)$.}
\label{fig-findcycle}
\end{figure}

For each piece $P$ in pseudo GMST $(\hat{H}_X,{\cal D},{\cal W}_X)$, let 
$\tilde{D}$ be the crest separator on the edge between $P$ and its parent node and 
let ${\cal D}_P$ be the set of crest separators on an edge between $P$ and a child 
node of $P$ in $T_{\cal D}$. From Property (A) and Property (V), any shortest path 
represented by an edge in $\upDDG(D)$ or $\lowDDG(D)$, $D\in {\cal D}$, does not 
contain any vertex of $V(\hat{H}_X)\setminus V(H_X)$. Therefore, $C_u$ for every 
$u$ in $L$ can be partitioned into subpaths such that each subpath is either 
entirely in $P(L)$ or is represented by an edge in $\upDDG(\tilde{D})$ 
or $\lowDDG(D)$, $D\in {\cal D}_P$. Let $P^*$ be the weighted graph consisting of 
the edges of $P(L)$, $\upDDG(\tilde{D})$ and $\lowDDG(D)$, $D\in {\cal D}_P$. 
Notice that for every edge $e$ in $\upDDG(\tilde{D})$ and $\lowDDG(D)$, 
$D\in {\cal D}_P$, a shortest path represented by $e$ is also computed. Then 
it is known (Section 2.2 in \cite{BSW13}) that a $C_y$ can be computed in
$O(t(P^*)\log |V(L)|+|V(C_y)|)=O(t(P^*)\log k+|V(C_y)|)$ time, where $t(P^*)$ is 
the time to find a shortest path $C_u$ in $P^*$ for any $u$ in $L$. For each edge 
of $P^*$, we can multiply the edge weight by 2 to make each edge weight a positive 
integer. By the algorithm in \cite{T99}, a shortest path $C_u$ can be computed in 
linear time, that is, $t(P^*)=O(|E(P^*)|)$. Therefore, from the fact that each 
crest separator has $O(k)$ vertices (Property III), it takes 
\begin{eqnarray*}
& & O((|E(P(L))|+|E(\upDDG(\tilde{D}))|+|\cup_{D\in \cal{D}_P} 
E(\lowDDG(D))|)\log k+|V(C_y)|)\\
& = &
O((|E(P)|+\Delta(P)k^2)\log k+|V(C_y)|)
\end{eqnarray*}
time to compute a minimum $(f_Z,f_{\overline X})$-separating cycle, where 
$\Delta(P)$ is the number of edges incident to $P$ in $T_{\cal D}$. 
If $C_y$ has length at most $k$ then $|V(C_y)|=O(k)$, otherwise $|V(C_y)|=O(m)$. 
Assume that $C_y$ for every $Z\in {\cal W}_X$ has length at most $k$. From 
Property (C), Property (IV) and $|V(T_{\cal D})|=|{\cal W}_X|=O(m/k)$, the time 
for computing the minimum $(f_Z,f_{\overline X})$-separating cycles for all crests
$Z\in {\cal W}_X$ is
\begin{eqnarray*}
\sum_{P\in V(T_{\cal D})} O((|E(P)|+\Delta(P)k^2)\log k+k) &=&
O((m+|{\cal W}_X|k+|{\cal W}_X|k^2)\log k+|{\cal W}_X|k)\\
&=&O(mk\log k).
\end{eqnarray*}
If there is a $C_y$ with length greater than $k$, it takes 
\[
O((|E(P)|+\Delta(P)k^2)\log k+m)=O(mk\log k)
\]
time to compute this $C_y$. After the first $C_y$ with length greater than $k$
is computed, the subroutine computes a cylinder minor in $O(m)$ time and
terminates. 

Summing up, the total time for Steps (1)-(5) is $O(|E(H_X)|k^2)$.
\end{proof}

Any two boundary cycles do not cross with each other because they do not share 
any common edge, each boundary cycle is the boundary of a face of $G_X$ and the
faces are disjoint. The height of any vertex in a boundary cycle is no smaller 
than the height of any $D\in {\cal D}$. Therefore from Property (B) and Property 
(V), a boundary cycle does not cross with a minimum separating cycle for a crest 
$Z$. However, a minimum separating cycle for one crest may cross with a minimum 
separating cycle for another crest. The next lemma gives a base for eliminating 
crossing separating cycles.

\begin {lemma}
Let $C_1$ be the minimum $(f_Z,f_{\overline X})$-separating cycle computed for 
crest $Z$ in Step (5). Let $C_2$ be the minimum 
$(f_{Z'},f_{\overline X})$-separating cycle computed for crest $Z'$ after $C_1$ 
in step (5). If $C_1$ and $C_2$ cross with each other, then there is a cycle $C$ 
such that $\ins(C)=\ins(C_1)\cup \ins(C_2)$ and the length of $C$ is the same as 
that of $C_2$.
\label{cycle-decomposition}
\end{lemma}
\begin {proof}
Assume that $P_1$ and $P_l$ in the underlying tree $T_{\cal D}$ contain $Z$ and 
$Z'$, respectively. Let $\{P_1,P_2\},\{P_2,P_3\},...,\{P_{l-1},P_l\}$ be the path 
between $P_1$ and $P_l$ in $T_{\cal D}$ and let the crest in $P_i$ be $Z_i$ for 
$1\leq i\leq l$ ($Z=Z_1$ and $Z'=Z_l$).  It is shown in \cite{KT13} (Lemma 30)
that if $Z_i\in \ins(C_1)$ then every $Z_j\in \ins(C_1)$ for $j<i$, and if 
$Z_i\in \ins(C_2)$ then every $Z_j\in \ins(C_2)$ for $j>i$. From this and the fact 
that $C_2$ is computed after $C_1$, which means $Z^\prime\notin \ins(C_1)$, there 
is a $Z_i$, $1\leq i<l$, such that $Z_i\in \ins(C_1)$ but $Z_j\not\in \ins(C_1)$ 
for $j>i$. 

Let $D$ be the crest separator for $Z_i$ and $Z_{i+1}$ in the pseudo GMST 
(see Figure~\ref{fig-intersect}). We assume without loss of generality that $Z$ 
(resp. $Z'$) is in the upper (resp. lower) piece by $D$. From Property (VI) and 
the fact that $Z_{i+1}\notin\ins(C_1)$, $C_1$ contains no (shortest path 
represented by) edge in $\lowDDG(D)$. Note that $D$ separates $Z$ from $Z^\prime$ 
in $\hat{H}_X$. From the fact that $C_1$ and $C_2$ cross with each other and the 
fact that $C_1$ contains no edge in $\lowDDG(D)$, $C_2$ has a subpath $C_2(u,v)$
represented by an edge $e=\{u,v\}$ in $\upDDG(D)$, where $u$ and $v$ are in 
different down paths of $D$. 

\begin{figure}[t]
\centerline{\includegraphics[scale=0.8]{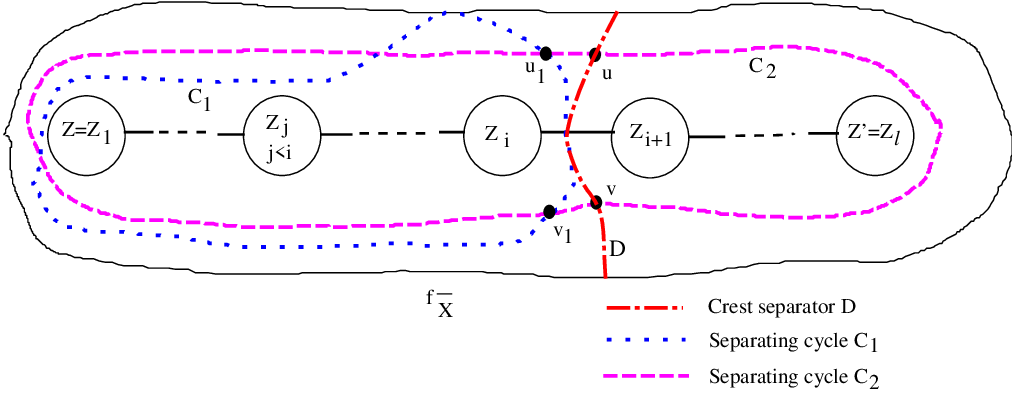}}
\caption{A pair of crossing minimum separating cycles $C_1$ and $C_2$.}
\label{fig-intersect}
\end{figure}

Since $C_1$ and $C_2$ cross with each other, 
they intersect at at least two vertices. Let $u_1$ and $v_1$ be the first vertex 
and last vertex at which cycle $C_2$ intersects cycle $C_1$, respectively, when 
we proceed on $C_2(u,v)$ from $u$ to $v$ (see Figure~\ref{fig-intersect}). Let 
$C_2(u_1,v_1)$ be the subpath of $C_2(u,v)$ connecting $u_1$ and $v_1$. Because
$C_1$ separates $Z_i$ from $Z_{i+1}$, it must contain a subpath $C_1(u_1,v_1)$ 
which connects $u_1$ and $v_1$ and intersects the ridge between $Z_i$ and 
$Z_{i+1}$. Let $C_2(u,u_1)$ (resp. $C_2(v,v_1)$) be the subpath of $C_2(u,v)$ 
between $u$ and $u_1$ (resp. $v$ and $v_1$). Then the length of $C_2(u,v)$ is the 
sum of the lengths of $C_2(u,u_1)$, $C_2(u_1,v_1)$ and $C_2(v,v_1)$. Let $D(u,v)$
be the subpath of $D$ between $u$ and $v$ that contains a top vertex of $D$. By 
Property (VI) and the fact that the closed walk formed by $C_2(u,u_1)$, 
$C_1(u_1,v_1)$, $C_2(v,v_1)$ and $D(u,v)$ does not separate $Z_i$ from 
$f_{\overline{X}}$, $\dist_D(u,v)$ is at most the length of the path concatenated
by $C_2(u,u_1)$, $C_1(u_1,v_1)$ and $C_2(v,v_1)$. From these and because the length 
of $C_2(u,v)$ is smaller than $\dist_D(u,v)$, the length of $C_2(u_1,v_1)$ is 
smaller than that of $C_1(u_1,v_1)$. From this, $Z\in \ins(C_2)$ because otherwise, 
we can replace $C_1(u_1,v_1)$ with $C_2(u_1,v_1)$ to get a separating cycle for $Z$ 
with a smaller length, a contradiction to that $C_1$ is a minimum separating cycle 
for $Z$.

For each connected region $R$ in $\ins(C_1)\setminus \ins(C_2)$, the boundary of 
$R$ consists of a subpath $C_1(R)$ of $C_1$ and a subpath $C_2(R)$ of $C_2$. 
The lengths of $C_1(R)$ and $C_2(R)$ are the same, otherwise, we can get a 
separating cycle for $Z$ or $Z'$ with length smaller than that of $C_1$ or $C_2$, 
respectively, a contradiction to that $C_1$ is a minimum separating cycle for $Z$ 
and $C_2$ is a minimum separating cycle for $Z'$.

We construct the cycle $C$ for the lemma as follows: Initially $C=C_2$. For 
every connected region $R$ in $\ins(C_1)\setminus \ins(C_2)$, we replace 
$C_2(R)$ by $C_1(R)$. Then $\ins(C)=\ins(C_1)\cup \ins(C_2)$ and has the same 
length as that of $C_2$.
\end {proof}

By applying Lemma \ref{cycle-decomposition1} repeatedly,
we get the next Lemma to eliminate a set of crossing minimum separating cycles.

\begin {lemma}
Let $C_1,C_2,...C_t$ be a set of crossing 
minimum separating cycles computed in Step (5) such that 
every $C_i$, $1<i\leq t$, is computed after $C_{i-1}$.
Then there is a cycle $C$ such that 
$\ins(C)=\ins(C_1)\cup\dots\ins(C_t)$ and the length of $C$ is the same as that 
of $C_t$.
\label{cycle-decomposition1}
\end{lemma}

Given a set of separating cycles computed in Steps (1) and (5) in Subroutine
Crest-Separator, our next job is to eliminate the crossing minimum separating
cycles and find a good separator $\AAA_X$ for ${\cal Z}_X$ and $\overline{X}$.
The next lemma shows how to do this (Step (6)).
\begin{lemma}
Given the set of separating cycles of length at most $k$ computed in Subroutine 
Crest-Separator, a good separator $\AAA_X$ for ${\cal Z}_X$ and $\overline{X}$ 
can be computed in $O(|E(\hat{H}_X)|)$ time.
\label{lem-intersect}
\end{lemma}
\begin{proof} 
Let $m=|E(H_X)|$ and ${\cal C}$ be the set of $(f_Z,f_{\overline{X}})$-separating 
cycles for all $Z\in {\cal Z}_X$. Let $\Gamma$ be the subgraph of $H_X$ induced by 
the edges of all cycles in ${\cal C}$. We orient each separating cycle 
$C\in {\cal C}$ such that $\ins(C)$ is on the right side when we proceed on $C$ 
following its orientation and give $C$ a distinct integer label $\lambda(C)$. We 
create a directed plane graph $\vec{\Gamma}$ with $V(\vec{\Gamma})=V(\Gamma)$ and 
\begin{eqnarray*}
E(\vec{\Gamma})=\{(u,v)_{\lambda(C)}&|& \{u,v\}\in E(C), C\in {\cal C}, \\
& & \mbox{ and the orientation of $C$ is from $u$ to $v$.}\}.
\end{eqnarray*}
Notice that if edge $\{u,v\}$ of $\Gamma$ appears in multiple cycles then
$\vec{\Gamma}$ may have parallel arcs from $u$ to $v$. For simplicity, we may 
use $(u,v)$ for arc $(u,v)_{\lambda(C)}$ when the label $\lambda(C)$ is not needed 
in the context. For each cycle $C\in {\cal C}$, we denote the corresponding 
oriented cycle in $\vec{\Gamma}$ by $\vec{C}$. The planar embedding of 
$\vec{\Gamma}$ is as follows: For each vertex $u$ in $\vec{\Gamma}$, the embedding 
of $u$ is the point of $\Sigma$ that is the embedding of vertex $u$ in $\Gamma$. 
For each edge $e=\{u,v\}$ in $\Gamma$, let $r_e$ be a region in $\Sigma$ such that 
$e\subseteq r_e$, $r_e$ does not have any point of $\Gamma$ other than $e$, and 
$r_e\cap r_{e'}=\emptyset$ for distinct edges $e$ and $e'$ of $\Gamma$ (see
Figure~\ref{orders} (a)). Each arc $\vec{e}=(u,v)_{\lambda(C)}$ in $\vec{\Gamma}$ 
is embedded as a segment in region $r_e$, $e=\{u,v\}$ (see Figure~\ref{orders} (b)). 
We further require the embeddings of arcs in $\vec{\Gamma}$ satisfying the 
{\em left-embedding property}: For each edge $e=\{u,v\}$ in $\Gamma$, if there is 
at least one arc from $u$ to $v$ and at least one arc from $v$ to $u$ in 
$\vec{\Gamma}$ then for any pair of arcs $\vec{e}=(u,v)$ and $\vec{e'}=(v,u)$, the 
embeddings of $\vec{e}$ and $\vec{e'}$ form an oriented cycle in $r_e$ such that 
none of $f_Z$ and $f_{\overline{X}}$ is on the left side when we proceed on the 
cycle following its orientation (see Figure~\ref{orders} (b)). $\Gamma$ has a 
face which includes $f_{\overline{X}}$ and we take this face as the outer face $f_0$
of $\vec{\Gamma}$. Since each edge of $H_X$ appears in at most one boundary cycle 
$C_Z$, there are $O(m)$ arcs $(u,v)_{\lambda(C_Z)}$ for 
$Z\in {\cal Z}_X\setminus {\cal W}_X$. Since $|{\cal W}_X|=O(m/k)$ and the 
minimum separating cycle $C$ for every $Z\in {\cal W}_X$ has at most $k$ edges, 
there are $O(|{\cal W}_X|k)=O(m)$ arcs $(u,v)_{\lambda(C)}$ in total for all
$Z\in {\cal W}_X$. Therefore, $\vec{\Gamma}$ can be computed in $O(m)$ time. For 
each face $f$ in $\Gamma$, let $\vec{E}(f)$ be the set of arcs $(u,v)$ and $(v,u)$ 
in $\vec{\Gamma}$ for each $\{u,v\}\in E(f)$ in $\Gamma$. 

\begin{figure}[t]
\centerline{\includegraphics[scale=0.6]{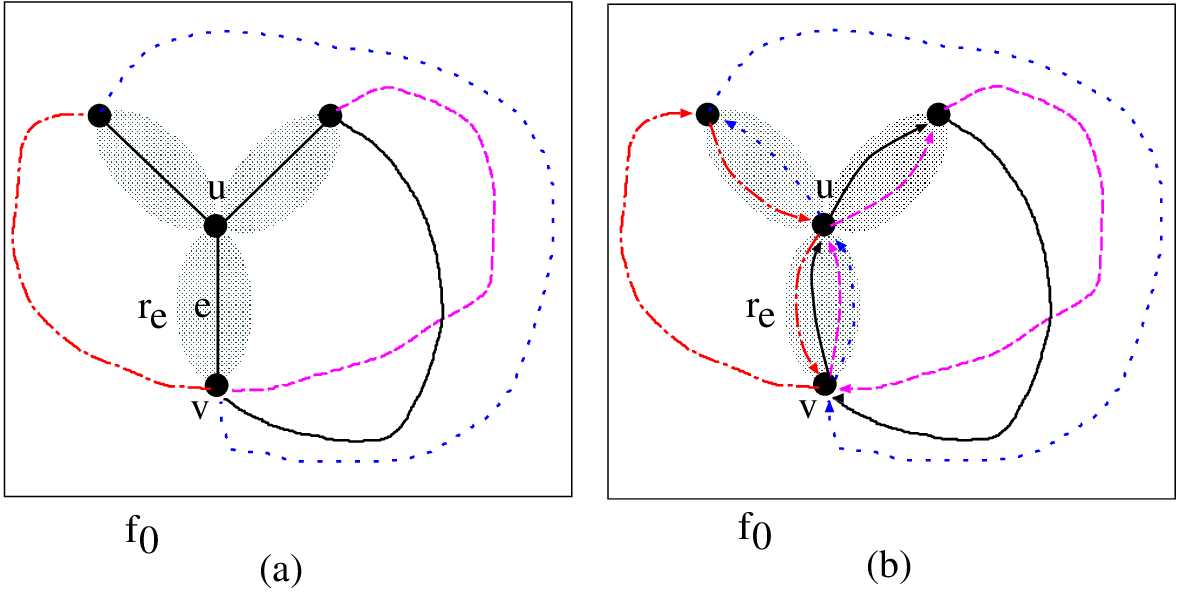}}
\caption{(a) embedding of $\Gamma$ and (b) embedding of $\vec{\Gamma}$.}
\label{orders}
\end{figure}

A search on arc $\vec{e}=(u,v)$ means that we proceed on arc $\vec{e}$ from $u$ 
to $v$. For each arc $\vec{e}=(u,v)_{\lambda(C)}$, we define its {\it next arc}
$\nx(\vec{e})=(v,w)_{\lambda(C)}$ and {\it previous arc} 
$\pv(\vec{e})=(t,u)_{\lambda(C)}$. For arc $\vec{e}=(u,v)_{\lambda(C)}$, let 
$C(\vec{e})$ be the oriented cycle that contains $\vec{e}$ and let 
$L(\vec{e})=\{\vec{h_1}=\nx(\vec{e}),..,\vec{h_t}\}$, $t\geq 1$, be the set of 
outgoing arcs from $v$ on the "left side" of $C(\vec{e})$. We assume that arc 
$\vec{h_i}$, $1\leq i\leq t$, in $L(\vec{e})$ is the $i$th outgoing arc from $v$ 
when we count the arcs incident to $v$ in the counter-clockwise order from 
$\nx(\vec{e})$ to $\vec{e}$. We define the {\em leftmost arc} from $\vec{e}$, 
denoted by $\lm(\vec{e})$, as the $\vec{h_i}\in L(\vec{e})$ with the largest $i$ 
such that $\pv(\vec{h_i})\in\ins(C(\vec{h_1}))\cup\dots\cup\ins(C(\vec{h}_{i-1}))$
(see Figure \ref{leftmost} for an example). For each arc 
$\vec{e}=(u,v)_{\lambda(C)}$, $\lm(\vec{e})$ can be found by checking the arcs 
in $L(\vec{e})$, starting from $\nx(\vec{e})$, in the counter-clockwise order 
they are incident to $v$. A search on a sequence of arcs $\vec{e_1},\vec{e_2},..$ 
is called a {\em leftmost search} if every $\vec{e}_{i+1}$ is $\lm(\vec{e_i})$ 
for $i\geq 1$.
\begin{figure}[t]
\centerline{\includegraphics[scale=0.7]{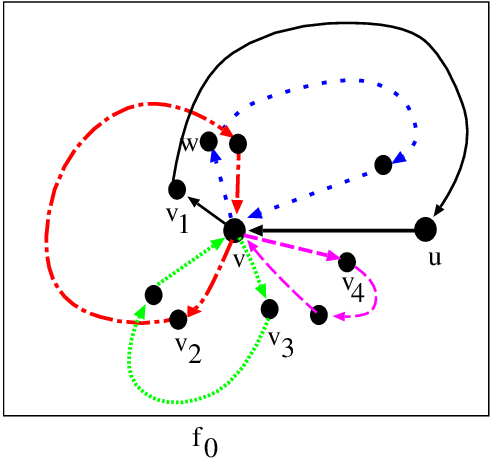}}
\caption{For arc $\vec{e}=(u,v)$, $\nx(\vec{e})=(v,v_1)$, 
$L(\vec{e})=\{(v,v_1),(v,v_2),(v,v_3),(v,v_4)\}$
and $\lm(\vec{e})=(v,v_3)$.}
\label{leftmost}
\end{figure}

By performing a leftmost search on arcs of $\vec{\Gamma}$, starting from an 
arbitrary arc in $\vec{E}(f_0)$, we can find a separating cycle $\vec{C_m}$ such 
that for any cycle $\vec{C}$, $C\in {\cal C}$, if 
$\ins(\vec{C_m})\cap \ins(\vec{C})\neq \emptyset$ then 
$\ins(\vec{C})\subseteq \ins(\vec{C_m})$. We call $\vec{C_m}$ a 
{\em maximal cycle}. According to Lemma~\ref{cycle-decomposition1} and the fact 
that every cycle $C\in {\cal C}$ has a length at most $k$, the length of 
$\vec{C_m}$ is at most $k$. 

After finding $\vec{C_m}$, we delete arcs in cycles $\vec{C}$ from $\vec{\Gamma}$ 
if $\ins(\vec{C})\subseteq \ins(\vec{C_m})$ to update $\vec{\Gamma}$. We continue 
the search on the updated $\vec{\Gamma}$ from an arbitrary arc in the updated 
$\vec{E}(f_0)$ until all arcs are deleted. Then for each $Z\in {\cal Z}_X$, there 
is a unique maximal cycle which separates $Z$ and $\overline{X}$. For every arc 
$\vec{e}$ in a maximal cycle $\vec{C_m}$, all the arcs in $\ins(\vec{C_m})$ are 
deleted after $\vec{C_m}$ is found. Each arc is counted $O(1)$ time in the 
computation for all leftmost arc searches. Therefore, the total time complexity 
of finding all the maximal cycles is $O(m)$.

For each maximal cycle $\vec{C_m}$ in $\vec{\Gamma}$ computed above, let $C_m$
be the cycle in $H_X$ consisting of edges corresponding to the arcs in 
$\vec{C_m}$. For each $Z\in {\cal Z}_X$, there is a cycle $C_m$ which separates 
$Z$ and $\overline{X}$. Let $A_Z$ be the edge subset induced by $C_m$. Then (1) 
$|\partial(A_Z)|\leq k$ (if $A_Z$ is induced by the boundary cycle $C_Z$ then 
$|\partial(A_Z)|=|C_Z|\leq k$, otherwise $A_Z$ is induced by a cycle $C_m$ of 
length at most $k$ as shown in Lemma~\ref{cycle-decomposition}, implying
$|\partial(A_Z)|\leq k$). Due to the way we find the maximal cycles above, (2)
for every $Z\in {\cal Z}_X$, there is exactly one subset $A_Z\in \AAA_X$
separating $Z$ and $\overline{X}$; and (3) for distinct $A_Z,A_{Z'}\in \AAA_X$,
$A_Z\cap A_{Z'}=\emptyset$. Therefore, $\AAA_X$ is a good-separator for
${\cal Z}_X$ and $\overline{X}$.
\end{proof}

We are ready to show Theorem~\ref{theo-3} which is re-stated below.
\begin{theorem} There is an algorithm which given a planar graph $G$ of $n$
vertices and an integer $k$, in $O(nk^2)$ time, either constructs a
branch-decomposition of $G$ with width at most $(2+\delta)k$ or a
$(k+1)\times \ceil{\frac{k+1}{2}}$ cylinder minor of $G$, where $\delta>0$
is a constant.
\label{theo-5}
\end{theorem}
\begin{proof}
First, as shown in Lemma~\ref{lem-intersect}, a good separator $\AAA_X$ for
${\cal Z}_X$ and $\overline{X}$ is computed by Subroutine Crest-Separator.
From this and as shown in the proof of Theorem~\ref{theo-4}, given a planar 
graph $G$ and integer $k$, our algorithm computes a branch-decomposition of $G$ 
with width at most $(2+\delta)k$ or a $(k+1)\times \ceil{\frac{k+1}{2}}$ 
cylinder minor of $G$.

Let $M,m_x,m$ be the numbers of edges in 
$G[\reach_{G|U}(\partial(U),d_2)],(G|\overline{X})|{\cal Z}_X, \hat{H}_X$, 
respectively. Then $m=O(m_x)$. By Lemmas~\ref{lem-all-cycle} and \ref{lem-intersect}, 
Subroutine Crest-Separator takes $O(mk^2)$ time.
For distinct level 1 nodes $X$ and $X'$, the edge sets of subgraphs 
$(G|\overline{X})|{\cal Z}_X$ and $(G|\overline{X'})|{\cal Z}_{X'}$ are disjoint. 
From this, $\sum_{X:\mbox{level 1 node}} m_x=O(M)$. Therefore, Step 2 of
Procedure Branch-Minor($G|U)$ takes 
$\sum_{X:\mbox{level 1 node}} O(m_xk^2)=O(Mk^2)$ time when Steps 2(b)(c) are
replaced by Subroutine Crest\_Separator.

The time for other steps in Procedure Branch-Minor($G|U$) is $O(M)$. The number
of recursive calls in which each vertex of $G|U$ is involved in the computation
of Step 2 is $O(\frac{1}{\alpha})=O(1)$. Therefore, we get an algorithm with
running time $O(nk^2)$.
\end{proof}

\section{Concluding remarks}
\label{sec-con}

If we modify the definition for $d_2$ in Section~\ref{sec-alg1} from 
$d_2=d_1+\ceil{\frac{k+1}{2}}$
to $d_2=d_1+(k+1)$, we get an algorithm which given a planar graph $G$
and integer $k>0$, in $\min\{O(n\log^3 n),O(nk^2)\}$ time either computes a 
branch-decomposition of $G$ with width at most $(3+\delta)k$, where $\delta>0$
is a constant, or a $(k+1)\times (k+1)$ cylinder minor (or grid minor). 
It is open whether there is an $O(n)$ time constant factor approximation 
algorithm for the branchwidth and largest grid (cylinder) minors. The algorithm
of this paper can be used to reduce the vertex cut set size in the recursive
division of planar graphs with small branchwidth in near linear time. It is 
interesting to investigate the applications of the algorithm in this paper.
Such applications include to improve the efficiency of graph decomposition 
based algorithms for problems in planar graphs.

\end{document}